%% file: netrate-ident-sc.tex
\documentclass{article}

\usepackage{times}
\usepackage{graphicx} 
\usepackage{subfigure}
\usepackage{natbib}
\usepackage{algorithm}
\usepackage{algorithmic}
\usepackage{hyperref}

\usepackage[accepted]{icml2014}

\usepackage{color}
\usepackage{url}
\usepackage{multirow}
\usepackage{enumerate}
\usepackage{xspace}
\usepackage{enumitem}
\usepackage{Definitions}

\newcommand{\denselistA}{ \itemsep -1pt\topsep-20pt\partopsep-20pt }

\newcommand{\hide}[1]{}
\newcommand{\xhdr}[1]{\vspace{1.7mm}\noindent{{\bf #1.}}}
\newcommand{\casc}{{\mathbf{t}}}
\newcommand{\alphs}{{\mathbf{A}}}

\newcommand{\netrate}{{\textsc{Net\-Rate}}\xspace}

\newcommand{\expo}{{\textsc{Exp}}\xspace}
\newcommand{\pow}{{\textsc{Pow}}\xspace}
\newcommand{\ray}{{\textsc{Ray}}\xspace}
\newcommand{\unobs}{{\infty}}

\graphicspath{{./FIG/}}

\begin{document}

\icmltitlerunning{Estimating Diffusion Network Structures}

\twocolumn[

\icmltitle{Estimating Diffusion Network Structures: Recovery Conditions, \\ Sample Complexity \& Soft-thresholding Algorithm}

\icmlauthor{Hadi Daneshmand$^1$}{hadi.daneshmand@tue.mpg.de}

\icmlauthor{Manuel Gomez-Rodriguez$^1$}{manuelgr@tue.mpg.de}

\icmlauthor{Le Song$^2$}{lsong@cc.gatech.edu}

\icmlauthor{Bernhard Sch\"{o}lkopf$^1$}{bs@tue.mpg.de}

\icmladdress{$^1$MPI for Intelligent Systems and $^2$Georgia Institute of Technology}

\icmlkeywords{inferring networks, diffusion}

\vskip 0.3in
]

\begin{abstract}
\input{000abstract}
\end{abstract}

\vspace{-3mm}
\section{Introduction}
\vspace{-2mm}

\setlength{\abovedisplayskip}{4pt}
\setlength{\abovedisplayshortskip}{1pt}
\setlength{\belowdisplayskip}{4pt}
\setlength{\belowdisplayshortskip}{1pt}
\setlength{\jot}{3pt}

\setlength{\textfloatsep}{3ex}

\label{sec:intro}
\input{010intro}

\vspace{-3mm}
\section{Continuous-Time Diffusion Model}
\label{sec:model}
\vspace{-2mm}
\input{020model}

\vspace{-3mm}
\section{Consistency}
\label{sec:consistency}
\vspace{-2mm}
\input{030consistency}
\vspace{-3mm}
\section{Recovery Conditions}
\label{sec:conditions}
\vspace{-2mm}
\input{040results}
\vspace{-3mm}
\section{Efficient soft-thresholding algorithm}
\label{sec:algorithm}
\vspace{-2mm}
\input{050algorithm}

\vspace{-3mm}
\section{Experiments}
\label{sec:experiments}
\vspace{-2mm}
\input{060experiments}

\vspace{-3mm}
\section{Conclusions}
\label{sec:conclusions}
\vspace{-2mm}
\input{070conclusions}

\vspace{-3mm}
\section*{Acknowledgement}
\vspace{-2mm}
This research was supported in part by NSF/NIH BIGDATA 1R01GM108341-01, NSF IIS1116886, and a Raytheon faculty fellowship to L. Song.

\bibliographystyle{icml2014}
\bibliography{refs}

\clearpage
\newpage

\begin{appendix}
\input{080appendix}

\end{appendix}

\end{document}

%% file: 000abstract.tex
Information spreads across social and technological networks, but often the network structures are hidden from us and we only observe the traces left by the diffusion 
processes, called \emph{cascades}. Can we recover the hidden network structures from these observed cascades? What kind of cascades and how many cascades 
do we need? Are there some network structures which are more difficult than others to recover? Can we design efficient inference algorithms with pro\-va\-ble 
guarantees?

Despite the increasing availability of cascade\- data and methods for inferring networks from these data, a thorough theo\-re\-ti\-cal un\-ders\-tan\-ding of the above questions 
remains largely unexplored in the literature. In this paper, we investigate the network structure inference pro\-blem for a general fami\-ly of continuous-time diffu\-sion models 
using an $\ell_1$-regularized likelihood ma\-xi\-mi\-zation framework. 
We show that, as long as the cascade sampling process satisfies a na\-tu\-ral inco\-herence condition, our framework can recover the correct network structure with high 
pro\-ba\-bi\-li\-ty if we observe $O(d^3 \log N)$ cascades, where $d$ is the maximum number of parents of a node and $N$ is the total number of nodes. Moreover, we develop a 
simple and efficient soft-thresholding inference algorithm, which we use to illustrate the consequences of our theoretical results, and show that our framework 
outperforms other alternatives in practice.

%


%% file: 010intro.tex
Diffusion of information, behaviors, diseases, or more ge\-ne\-rally, \emph{contagions} can be naturally mo\-deled as a stochastic process that occur over the edges
of an underlying network~\cite{rogers95diffusion}.
In this scenario, we often observe the temporal traces that the diffusion generates, called \emph{cascades}, but the edges of the network that gave rise to the diffu\-sion 
remain unobservable~\cite{adar05epidemics}.
For example, blogs or media sites often publish a new piece of information without explicitly citing their sources.
Marketers may note when a social media user decides to adopt a new behavior but cannot tell which neighbor in the social network influenced them to do so.
Epidemiologist observe\- when a person gets sick but usually cannot tell who infected her.
In all these cases, given a set of cascades and a diffusion model, the network inference problem consists of inferring the edges (and model parameters) of
the unobserved underlying network~\cite{phdmanuelgr2013}. 

The network inference problem has attracted significant attention in recent years~\cite{saito2009learning, manuel10netinf, manuel11icml, snowsill2011kdd, nandu12nips},
since it is essential to reconstruct and predict the paths over which information can spread, and to ma\-xi\-mize sales of a product or stop infections.
Most previous work has focused on de\-ve\-lo\-ping network inference algorithms and evaluating their performance experimentally on different synthetic and real networks, and a 
ri\-go\-rous theoretical analysis of the problem has been missing.
However, such analysis is of outstanding interest since it would enable us to answer many fundamental open questions. For example,
which conditions are sufficient to gua\-ran\-tee that we can recover a network given a large number of cascades?
If these conditions are satisfied, how many cascades are sufficient to infer the network with high pro\-ba\-bi\-li\-ty?
Until recently, there has been a paucity of work along this direction~\cite{netrapalli12,abrahaotrace} which provide only partial views of the problem. None of them is able to identify the recovery condition relating to the interaction between the network structure and the cascade sampling process, which we will make precise in our paper.

\xhdr{Overview of results} We consider the network inference problem under the continuous-time diffusion model recently introduced by~\citet{manuel11icml}.
%
We identify a natural incoherence condition for such a model which depends on both the network structure, the diffu\-sion parameters and the sampling process of the cascades. This condition 
captures the intuition that we can recover the network structure if the co-occurrence of a node and its non-parent nodes is small in the cascades. Furthermore, we show that, if this condition holds 
for the po\-pu\-lation case, we can recover the network structure \-using\- an $\ell_1$-regularized ma\-xi\-mum like\-li\-hood estimator and $O(d^3 \log N)$ cascades, and the probability of suc\-cess is 
approa\-ching 1 in a rate exponential in the number of cascades. Importantly, if this condition also holds for the finite sample case, then the guarantee can be improved to $O(d^2 \log N)$ 
cascades.
Beyond theoretical results, we also propose a new, efficient and simple pro\-xi\-mal gradient algorithm to solve the $\ell_1$-regularized maxi\-mum likelihood estimation. The algorithm is especially 
well-suited for our problem since it is highly scalable and naturally finds sparse estimators, as desired, by using soft-thresholding.
Using this algorithm, we perform various experiments illus\-tra\-ting the consequences of our theoretical results and demonstrating that it typically outperforms other state-of-the-art algorithms.

{\bf Related work.}
%
\citet{netrapalli12} propose a maximum likelihood network inference method for a variation of the discrete-time independent cascade model~\cite{kempe03maximizing} and show that,
for general net\-works satisfying a \emph{correlation decay}, the estimator recovers the network structure given $O(d^2 \log N)$ cascades, and the probability of success is approaching 1 in
a rate exponential in the number of cascades. The rate they obtained is on a par with our results. However, their discrete diffusion model is less realistic in practice, and the correlation decay
condition is rather restricted: essentially, on average each node can only infect one single node per cascade.
Instead, we use a general continuous-time diffusion model~\cite{manuel11icml}, which has been extensively validated in real diffusion data and extended in various ways by different
authors~\cite{wang2012feature,nandu12nips,du13aistats}.

%

\citet{abrahaotrace} propose a simple network in\-fe\-rence\- method, First-Edge, for a slightly different continuous-time independent cascade model~\cite{manuel10netinf}, and show that, for
general networks, if the cascade sources are chosen uniformly at random, the algorithm needs $O(N d \log N)$ cascades to recover the network structure and the probability of success is
approaching 1 only in a rate polynomial in the number of cascades. Additionally, they study trees and bounded-degree networks and show that, if the cascade sources are chosen
uniformly at random, the error decreases polynomially as long as $O(\log N)$ and $\Omega(d^9 \log^2 d \log N)$ cascades are recorded respectively.
In our work, we show that, for general networks satisfying a natural incoherence condition, our method outperforms the First-Edge algorithm and the algorithm for bounded-degree
networks in terms of rate and sample complexity.

\citet{gripon2013reconstructing} propose a network inference method for unordered cascades, in which nodes that are infected together in the same cascade are connected by a path
containing exactly the nodes in the trace, and give necessary and sufficient conditions for network inference.
However, they consider a restrictive, unrealistic scenario in which cascades are all three nodes long.

%% file: 020model.tex
%
\begin{figure}[t] %
  \vspace{-3mm}
\centering
  \includegraphics[width=0.45\textwidth]{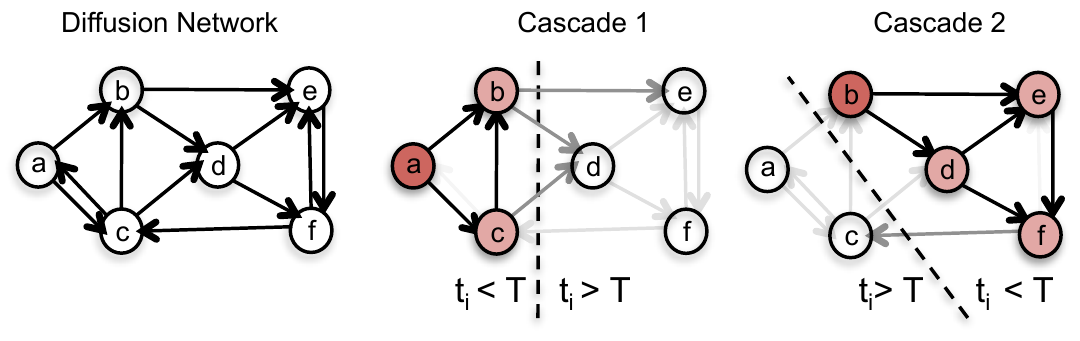}
  \vspace{-4mm}
  \caption{The diffusion network structure (left) is unknown and we only observe cascades, which are $N$-dimensional vectors recording the times when nodes get
  infected by contagions that spread (right).
  Cascade 1 is $(t_a, t_b, t_c, \infty, \infty, \infty)$, where $t_a < t_c < t_b$, and cascade 2 is $(\infty, t_b, \infty, t_d, t_e, t_f)$, where $t_b < t_d < t_e < t_f$.
  Each cascade contains a source node (dark red), drawn from a source distribution $\PP(s)$, as well as infected (light red) and uninfected (white) nodes,
  and it provides information on black and dark gray edges but does not on light gray edges.
  %
  } \label{fig:model}
  \vspace{-3mm}
\end{figure}

In this section, we revisit the continuous-time generative model for cascade data introduced by~\citet{manuel11icml}. The model associates each edge $j\rightarrow i$ with a transmission
function, $f(t_i|t_j ; \alpha_{ji})=f(t_i-t_j ; \alpha_{ji})$, a density over time parameterized by $\alpha_{ji}$. This is in contrast to previous discrete-time models which associate each edge with
a fixed infection probability~\cite{kempe03maximizing}. Moreover, it also differs from discrete-time models in the sense that events in a cascade are not generated iteratively in rounds, but
event timings are sampled directly from the transmission functions in the continuous-time model.
\begin{table*}[t]
    \vspace{-4mm}
    \caption{Functions.}
    \label{tab:functions}
    \begin{center}
    \begin{tabular*}{\textwidth}{@{\extracolsep{\fill}} l c c}
	  \hline
          \textbf{Function} & \textbf{ Infected node ($t_i < T$)} & \textbf{ Uninfected node ($t_i > T$)} \\
	$g_i(\casc; \boldsymbol{\alpha})$ &
	$\log h(\mathbf{t};\boldsymbol{\alpha})+\sum_{j :
	t_j < t_i} y(t_i | t_j ; \alpha_{j})$ & $\sum_{j : t_j < T} y(T | t_j ; \alpha_{j})$ \\

	$[\nabla y(\casc ; \boldsymbol{\alpha})]_k$ &
	$-y^{\prime}(t_i|t_k;\alpha_k)$ &
	$-y^{\prime}(T|t_k;\alpha_k)$ \\

	$[D(\casc ; \boldsymbol{\alpha})]_{kk}$ &
	$ - y^{\prime \prime} (t_i|t_k;\alpha_k) -
	h(\mathbf{t};\boldsymbol{\alpha})^{-1}H^{\prime \prime}(t_i|t_k;\alpha_k)
	$ &
	$ - y^{\prime \prime} (T|t_k;\alpha_k)$ \\
	\hline
    \end{tabular*}
    \end{center}
	\vspace{-4mm}
\end{table*}

\vspace{-3mm}
\subsection{Cascade generative process} \label{sec:generative-model}
\vspace{-2mm}

%
Given a \emph{directed} contact network, $\Gcal = (\Vcal,\Ecal)$ with $N$ nodes, the process begins with an infected source node, $s$, initially adopting certain \emph{contagion}
at time zero, which we draw from a source distribution $\PP(s)$. The contagion is transmitted from the source along her out-going edges to her direct neighbors. Each transmission
through an edge entails a \emph{random} transmission time, $\tau=t_j-t_j$, drawn from an associated transmission function $f(\tau ; \alpha_{ji})$.
We assume transmission times are independent, possibly distributed differently across edges, and, in some cases, can be arbitrarily large,  $\tau \rightarrow \infty$.
Then, the infected neighbors transmit the contagion to their respective neighbors, and the process continues.
We assume that an infected node remains infected for the entire diffusion process. Thus, if a node $i$ is infected by multiple neighbors, only the neighbor that first infects node $i$ will be the \emph{true
parent}. 
Figure~\ref{fig:model} illustrates the process.

%
Observations from the model are recorded as a set $C^{n}$ of cascades $\{\casc^1,\ldots,\casc^{n}\}$.
Each cascade $\casc^c$ is an $N$-dimensional vector $\casc^c:=(t^c_1,\ldots,t^c_N)$ recording when nodes are infected,
$t^c_k\in [0,T^c]\cup\{\unobs\}$. Symbol $\unobs$ labels nodes that are not infected during observation window $[0,T^c]$ --
it does not imply they are never infected. The `clock' is reset to 0 at the start of each cascade. We assume $T^c = T$ for all cascades;
the results generalize trivially.
%

%

\vspace{-3mm}
\subsection{Likelihood of a cascade}
\vspace{-2mm}
%
\citet{manuel11icml} showed that the likelihood of a cascade $\casc$ under the continuous-time independent cascade model is
\begin{multline} \label{eq:loglikelihood}
	f(\casc ; \alphs) = \prod_{t_i \leq T} \prod_{t_m > T} S(T | t_i ; \alpha_{i m}) \times \\
	\prod_{k : t_k < t_i} S(t_i | t_k ; \alpha_{ki}) \sum_{j : t_j < t_i}
	H(t_i | t_j ; \alpha_{ji}),
\end{multline}
where $\alphs=\cbr{\alpha_{ji}}$ denotes the collection of parameters, $S(t_i | t_j ; \alpha_{ji}) = 1-\int_{t_j}^{t_i} f(t | t_j ; \alpha_{ji})\, d t$ is the survival function and
$H(t_i | t_j ; \alpha_{ji})=f(t_i | t_j ; \alpha_{ji}) / S(t_i | t_j ; \alpha_{ji})$ is the hazard function.
The survival terms in the first line account for the probability that uninfected nodes survive to all infected nodes in the cascade up to $T$ and the survival and hazard terms
in the second line account for the likelihood of the infected nodes.
Then, assuming cascades are sampled independently, the likelihood of a set of cascades is the product of the likelihoods of individual cascades given by
Eq.~\ref{eq:loglikelihood}. For notational simplicity, we define $y(t_i | t_k ; \alpha_{ki}) := \log S(t_i | t_k ; \alpha_{k i})$, and
$h(\boldsymbol{t};\boldsymbol{\alpha}_i) := \sum_{k:t_k \leq t_i} H(t_i|t_k;\alpha_{ki})$ if $t_i \leq T$ and 0 otherwise.
%

\vspace{-3mm}
\section{Network Inference Problem}
\vspace{-2mm}
Consider an instance of the continuous-time diffusion model defined above with a contact network $\Gcal^{*} = (\Vcal^{*}, \Ecal^{*})$ and associated parameters $\cbr{\alpha^{*}_{ji}}$.
We denote the set of pa\-rents of node $i$ as $\Ncal^{-}(i) = \{ j \in \Vcal^{*} : \alpha^{*}_{ji} > 0 \}$ with cardinality $d_i = |\Ncal^{-}(i)|$ and the minimum positive transmission rate as
$\alpha^*_{\min,i} = \min_{j \,:\, \alpha^{*}_{ji}>0} \alpha^{*}_{ji}$.
Let $C^n$ be a set of $n$ cascades sampled from
the model, where the source $s \in \Vcal^*$ of each cascade is drawn from a source distribution $\PP(s)$. Then, the network inference problem consists of fin\-ding the directed edges and the
asso\-cia\-ted parameters using only the temporal information from the set of cascades $C^n$.

This problem has been cast as a maximum likelihood estimation problem~\cite{manuel11icml}
\begin{equation}
	\label{eq:opt-problem}
	\begin{array}{ll}
		\mbox{minimize$_{\boldsymbol{\alphs}}$} & - \frac{1}{n} \sum_{c \in C^n} \log f(\casc^c;\alphs) \\
		\mbox{subject to} & \alpha_{ji} \geq 0,\, i, j=1,\ldots,N, i \neq j,
	\end{array}
\end{equation}
where the inferred edges in the network correspond to those pairs of nodes with non-zero parameters,~\ie~$\hat{\alpha}_{ji} > 0$.

In fact, the problem in Eq.~\ref{eq:opt-problem} decouples into a set of independent smaller subproblems, one per node, where we infer the parents of each node and the parameters
asso\-cia\-ted with these incoming edges. Without loss of generality, for a particular node $i$, we solve the problem
\begin{equation}
	\begin{array}{ll}
		\mbox{minimize$_{\boldsymbol{\alpha}_i}$} & \ell^{n}(\boldsymbol{\alpha}_i) \\
		\mbox{subject to} & \alpha_{ji} \geq 0,\, j=1,\ldots,N, i \neq j,
	\end{array}
	\label{eq:opt-problem-one-node}
\end{equation}
where $\boldsymbol{\alpha}_i := \{\alpha_{ji}\,|\, j=1,\ldots,N, i \neq j\}$ are the relevant variables, and $\ell^{n}(\boldsymbol{\alpha}_i)  = - \frac{1}{n} \sum_{c \in C^n} g_i(\casc^c; \boldsymbol{\alpha}_i)$
corres\-ponds to the terms in Eq.~\ref{eq:opt-problem} involving $\boldsymbol{\alpha}_i$ (also see Table~\ref{tab:functions} for the definition of $g(\, \cdot \,; \boldsymbol{\alpha}_i)$).
In this subproblem, we only need to consider a super-neighborhood $\Vcal_i = \Rcal_i \cup \Ucal_i$ of $i$, with cardinality $p_i = |\Vcal_i| \leq N$, where $ \Rcal_i$ is the set of upstream nodes from 
which $i$ is reachable, $\Ucal_i$ is the set of nodes which are reachable from at least one node $j\in \Rcal_i$. 
Here, we consider a node $i$ to be reachable from a node $j$ if and only if there is a directed path from $j$ to $i$. 
We can skip all nodes in $\Vcal \backslash \Vcal_i$ from our analysis because they will never be infected in a cascade before $i$, and thus, the maximum likelihood estimation of the associated transmission 
rates will always be zero (and correct).


Below, we show that, 
as $n \rightarrow \infty$, the solution, $\hat{\boldsymbol{\alpha}}_i$, of the problem in Eq.~\ref{eq:opt-problem-one-node} is a consistent estimator of the true parameter $\boldsymbol{\alpha}^{*}_i$. 
%
However, it is not clear whether it is possible to recover the true network structure with this approach given a finite amount of cascades and, if so, how many cascades are needed.
We will show that by adding an $\ell_1$-regularizer to the objective function and solving instead the following optimization problem
\begin{equation}
	\label{eq:opt-problem-regularized-one-node}
	\begin{array}{ll}
		\mbox{minimize$_{\boldsymbol{\alpha}_i}$} & \ell^{n}(\boldsymbol{\alpha}_i) + \lambda_{n} ||\boldsymbol{\alpha}_i||_{1} \\
		\mbox{subject to} & \alpha_{ji} \geq 0,\, j=1,\ldots,N, i \neq j,
	\end{array}
\end{equation}
we can provide finite sample guarantees for recovering the network structure (and parameters).
Our analysis also shows that by se\-lec\-ting an appropriate value for the re\-gu\-la\-ri\-zation parameter $\lambda_n$, the solution of Eq.~\ref{eq:opt-problem-regularized-one-node} successfully
recovers the network structure with probability approa\-ching 1 exponentially fast in $n$.

In the remainder of the paper, we will focus on estimating the parent nodes of a particular node $i$. For simplicity, we will use $\boldsymbol{\alpha} = \boldsymbol{\alpha}_i$, $\alpha_{j} = \alpha_{ji}$, $\Ncal^{-} = \Ncal^{-}(i)$, $\Rcal = \Rcal_i$, $\Ucal = \Ucal_i$, $d = d_i$, $p_i = p$ and $\alpha^*_{\min}= \alpha^*_{\min,i}$.

%% file: 030consistency.tex
\emph{Can we recover the hidden network structures from the observed cascades?} The answer is yes. We will show this by proving that the estimator provided by Eq.~\ref{eq:opt-problem-one-node} is consistent, meaning that as the number of cascades goes to {\bf infinity}, we can always recover the true network structure.

More specifically, \citet{manuel11icml} showed that the network inference problem defined in Eq.~\ref{eq:opt-problem-one-node} is convex in $\boldsymbol{\alpha}$ if the survival functions are log-concave
and the hazard functions are concave in 
$\boldsymbol{\alpha}$. Under these conditions, the Hessian matrix, $\mathcal{Q}^n = \nabla^2 \ell^n(\boldsymbol{\alpha})$,
can be expressed as the sum of a nonnegative diagonal matrix $\boldsymbol{D}^n$ and the outer product of a matrix $\boldsymbol{X}^n (\boldsymbol{\alpha})$ with itself,~\ie,
\begin{equation} \label{eq:consistency-hessian-matrix-factorization}
	\hspace{-2mm}
	\mathcal{Q}^n =  \boldsymbol{D}^n(\boldsymbol{\alpha}) + \smallfrac{1}{n} \boldsymbol{X}^n(\boldsymbol{\alpha})[\boldsymbol{X}^n(\boldsymbol{\alpha})]^\top.
\end{equation}
Here the diagonal matrix $\boldsymbol{D}^n(\boldsymbol{\alpha}) = \smallfrac{1}{n} \sum_{c} \boldsymbol{D}(\casc^c ; \boldsymbol{\alpha})$ is a sum over a set of diagonal matrices $\boldsymbol{D}(\casc^c ; \boldsymbol{\alpha})$, one for each cascade $c$ (see Table~\ref{tab:functions} for the definition of its entries); and $\boldsymbol{X}^n(\boldsymbol{\alpha})$ is the Hazard matrix%
\begin{equation}
	\boldsymbol{X}^n(\boldsymbol{\alpha}) = \left[\boldsymbol{X}(\casc^1 ; \boldsymbol{\alpha}) \,|\, \boldsymbol{X}(\casc^2 ; \boldsymbol{\alpha}) \,|\, \ldots \,|\, \boldsymbol{X}(\casc^n ; \boldsymbol{\alpha})\right],
\end{equation}
with each column $\boldsymbol{X}(\casc^c ; \boldsymbol{\alpha}) := h(\boldsymbol{t}^c;\boldsymbol{\alpha})^{-1} \nabla_{\boldsymbol{\alpha}} h(\boldsymbol{t}^c;\boldsymbol{\alpha})$.
Intuitively, the Hessian matrix captures the co-occurrence information of nodes in cascades.
Then, we can prove
\begin{theorem} \label{th:consistency}
	\vspace{-2mm}
	If the source probability $\PP(s)$ is strictly positive for all $s \in \Rcal$, then, the maximum likelihood estimator $\hat{\boldsymbol{\alpha}}$ given by the solution of
	Eq.~\ref{eq:opt-problem-one-node} is consistent.
	\vspace{-2mm}
\end{theorem}
\begin{proof}
We check the three criteria for consistency: continuity, compactness and identification of the objective function~\cite{newey94}. Continuity is ob\-vious. For compactness, since $L\rightarrow -\infty$
for both $\alpha_{ij}\rightarrow 0$ and $\alpha_{ij}\rightarrow \infty$ for all $i,j$ so we lose nothing imposing upper and lower bounds thus restricting to a compact subset. For the identification
condition, $\boldsymbol{\alpha} \neq \boldsymbol{\alpha}^* \Rightarrow \ell^n(\boldsymbol{\alpha}) \neq \ell^n(\boldsymbol{\alpha}^*)$,  we use Lemma~\ref{lem:hessian-positive-definite} and~\ref{lem:nonsingularity} (refer to Appendices~\ref{app:proof-hessian-positive-definite} and~\ref{app:proof-nonsingularity}), which establish that $\boldsymbol{X}^n(\boldsymbol{\alpha})$ has full
row rank as $n\rightarrow \infty$, and hence $\mathcal{Q}^n$ is positive definite.
\vspace{-4mm}
\end{proof}

%% file: 040results.tex
In this section, we will find a set of sufficient conditions on the diffusion model and the cascade sampling process under which we can recover the network structure from 
{\bf finite samples}. These results allow us to address two questions:
\begin{itemize}[noitemsep,nolistsep]
  \item \emph{Are there some network structures which are more diffi\-cult than others to recover?}
  \item \emph{What kind of cascades are needed for the network structure recovery?}
\end{itemize}
The answers to these questions are intertwined. The diffi\-cul\-ty of finite-sample recovery depends crucially on an incoherence condition which is a function of both network structure, 
parameters of the diffusion model and the cascade sampling process. 
Intuitively, the sources of the cascades in a diffusion network have to be chosen in such a way that nodes without parent-child relation should co-occur less often compared to 
nodes with such relation. Many commonly used diffusion models and network structures can be naturally made to satisfy this condition.

More specifically, we first place two conditions on the Hessian of the population log-likelihood, $\mathbb{E}_c\sbr{\ell^n(\boldsymbol{\alpha}) } = \mathbb{E}_c \sbr{\log g(\casc^{c}; \boldsymbol{\alpha})}$, where the expectation here is taken over the distribution $\PP(s)$ of the source nodes, and the density $f(\casc^{c} | s)$ of the cascades $\casc^{c}$ given a source
node $s$. 
In this case, we will further denote the Hessian of $\mathbb{E}_c \sbr{\log g(\casc^{c}; \boldsymbol{\alpha})}$ evaluated at the true model parameter $\boldsymbol{\alpha}^*$ as $\Qcal^*$.
Then, we place two conditions on the Lipschitz continuity of $\boldsymbol{X}( \casc^c ; \boldsymbol{\alpha})$, and the boundedness of $\boldsymbol{X}( \casc^c ; \boldsymbol{\alpha}^*)$
and $\nabla g(\mathbf{t}^c;\boldsymbol{\alpha}^*)$ at the true model parameter $\boldsymbol{\alpha}^*$.
For simpli\-city, we will denote the subset of indexes associated to node $i$'{}s true parents as $S$, and its complement as $S^c$. Then, we use $\mathcal{Q}^{*}_{SS}$ to denote
the sub-matrix of $\mathcal{Q}^{*}$ indexed by $S$ and $\boldsymbol{\alpha}^{*}_{S}$ the set of parameters indexed by $S$.
%

{\bf Condition 1 (Dependency condition):} There exists constants $C_{min} > 0$ and $C_{max} > 0$ such that $\Lambda_{min}\left(\mathcal{Q}^{*}_{SS}\right) \geq C_{min}$ and
$\Lambda_{max}\left(\mathcal{Q}^{*}_{SS}\right) \leq C_{max}$ where $\Lambda_{min}(\cdot)$ and $\Lambda_{max}(\cdot)$ return the leading and the bottom eigenvalue of its
argument respectively. This assumption ensures that two connected nodes co-occur reasonably frequently in the cascades but are not deterministically related.

{\bf Condition 2 (Incoherence condition):} There exists $\varepsilon \in (0, 1]$ such that $||| \mathcal{Q}^{*}_{S^c S} \left(\mathcal{Q}^{*}_{SS}\right)^{-1} |||_{\infty} \leq 1 - \varepsilon$,
where $||| A |||_{\infty} = \max_{j} \sum_{k} |A_{jk}|$. This assumption captures the intuition that, node $i$ and any of its neighbors should get infected together in a cascade more
often than node $i$ and any of its non-neighbors.

{\bf Condition 3 (Lipschitz Continuity):} For any feasible cascade $\casc^c$, the Hazard vector $\boldsymbol{X}(\casc^c ; \boldsymbol{\alpha})$ is Lip\-schitz con\-ti\-nuous in the domain
$\{ \boldsymbol{\alpha} : \boldsymbol{\alpha}_S \geq \alpha^*_{\min}/2 \}$,
\begin{equation*}
	\|\boldsymbol{X}(\casc^c ; \boldsymbol{\beta}) - \boldsymbol{X}( \casc^c ; \boldsymbol{\alpha})\|_2 \leq k_1 \|\boldsymbol{\beta} - \boldsymbol{\alpha}\|_2,
\end{equation*}
where $k_1$ is some positive constant. As a consequence, the spectral norm of the difference, $n^{-1/2} (\boldsymbol{X}^n(\boldsymbol{\beta}) - \boldsymbol{X}^n(\boldsymbol{\alpha}))$,
is also bounded (refer to appendix~\ref{app:proof-bounded_hazard}),~\ie,
\begin{align}
	|||n^{-1/2}\big(\boldsymbol{X}^n(\boldsymbol{\beta})-\boldsymbol{X}^n(\boldsymbol{\alpha})\big)|||_2
	\leq 	k_1 \|\boldsymbol{\beta} -\boldsymbol{\alpha}\|_2.
	\label{eq:bounded_hazard}
\end{align}
Furthermore, for any feasible cascade $\casc^c$, $\boldsymbol{D}(\boldsymbol{\alpha})_{jj}$ is Lipschitz continuous for all $j \in \Vcal $,
\begin{equation*}
	|\boldsymbol{D}(\casc^c ; \boldsymbol{\beta})_{jj} - \boldsymbol{D}(\casc^c ; \boldsymbol{\alpha})_{jj}| \leq k_2 \|\boldsymbol{\beta} - \boldsymbol{\alpha}\|_2,
\end{equation*}
where $k_2$ is some positive constant.

{\bf Condition 4 (Boundedness):}  For any feasible cascade $\casc^c$, the absolute value of each entry in the gradient of its log-likelihood and in the Hazard vector, as evaluated at the true model parameter $\boldsymbol{\alpha}^*$, is bounded,
\begin{align*}
	\| \nabla g(\mathbf{t}^c;\boldsymbol{\alpha}^*) \|_{\infty} \leq k_3,\quad
	\|\boldsymbol{X}(\mathbf{t}^c;\boldsymbol{\alpha}^*)\|_{\infty} \leq k_4,
\end{align*}
where $k_3$ and $k_4$ are positive constants. Then the absolute value of each entry in the Hessian matrix $\mathcal{Q}^\ast$, is also bounded $|||\mathcal{Q}^*|||_{\infty} \leq k_5$.

{\bf Remarks for condition 1} As stated in Theorem~\ref{th:consistency}, as long as the source probability $\PP(s)$ is strictly positive for all $s \in \Rcal$, the maximum likelihood formulation
is strictly convex and thus there exists $C_{min} >0 $ such that $\Lambda_{min}\left(\mathcal{Q}^{*}\right) \geq C_{min}$.
%
%
Moreover, 
condition 4 implies that there exists $C_{max}>0$ such that $\Lambda_{max}\left(\mathcal{Q}^{*}\right) \leq C_{max}$.

{\bf Remarks for condition 2}
The incoherence condition depends, in a non-trivial way, on the network structure, diffusion parameters, observation window and source node distribution.
Here, we give some intuition by studying three small canonical examples.

First, consider the chain graph in Fig.~\ref{fig:chain} and assume that we would like to find the incoming edges to node $3$ when $T \rightarrow \infty$. Then, it is easy to show that the incoherence
condition is satisfied if $(P_0+P_1) / (P_0+P_1+P_2) < 1-\varepsilon$ and $P_0 / (P_0+P_1+P_2) < 1-\varepsilon$, where $P_i$ denotes the pro\-ba\-bi\-li\-ty of a node $i$ to be the source of a
cascade. Thus, for example, if the source of each cascade is chosen uniformly at random, the inequality is satisfied. Here, the incoherence condition depends on the source
node distribution.

Second, consider the directed tree in Fig.~\ref{fig:tree} and assume that we would like to find the incoming edges to node $0$ when $T \rightarrow \infty$. Then, it can be shown that the
incoherence condition is satisfied as long as (1) $P_1 > 0$, (2) ($P_2 > 0$) or ($P_5 > 0$ and $P_6 > 0$), and (3) $P_3 > 0$. As in the chain, the condition depends on the source node
distribution.

Finally, consider the star graph in Fig.~\ref{fig:star}, with exponential edge transmission functions, and assume that we would like to find the incoming edges to a leave node $i$ when $T < \infty$. 
Then, as long as the root node has a nonzero probability $P_0 > 0$ of being the source of a cascade, it can be shown that the incoherence condition reduces to the inequalities 
$\left(1-\frac{\alpha_{0j}}{\alpha_{0i}+\alpha_{0j}}\right) e^{-(\alpha_{0i}+\alpha_{0j}) T} + \frac{\alpha_{0j}}{\alpha_{0i}+\alpha_{0j}} < 1 - \varepsilon(1+e^{-\alpha_{0i} T})$, $j=1,\ldots,p\, :\, j \neq i$, 
which always holds for some $\varepsilon > 0$.
If $T \rightarrow \infty$, then the condition holds whenever $\varepsilon < \alpha_{0i}/(\alpha_{0i} + \max_{j : j \neq i} \alpha_{0j})$.
Here, the larger the ratio $\max_{j : j \neq i} \alpha_{0j} / \alpha_{0i}$ is, the smaller the maximum value of $\varepsilon$ for which the incoherence condition holds.
To summarize, as long as $P_0 > 0$, there is always some $\varepsilon > 0$ for which the condition holds, and such $\varepsilon$ value depends on the time
window and the parameters $\alpha_{0j}$.
\begin{figure}[t] %
	\vspace{-4mm}
\centering
  \subfigure[Chain]{\makebox[2cm][c]{\includegraphics[width=0.035\textwidth]{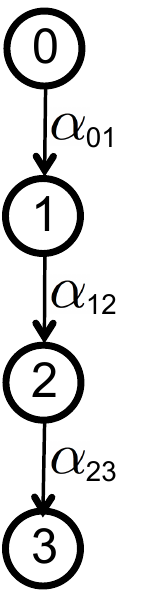} \label{fig:chain}}} \hspace{2mm}
  \subfigure[Tree]{\includegraphics[width=0.12\textwidth]{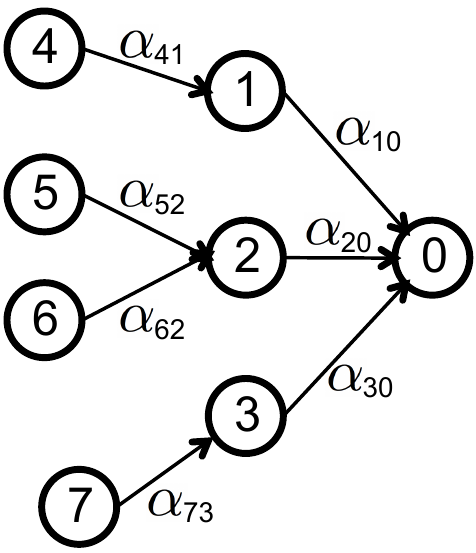} \label{fig:tree}} \hspace{2mm}
  \subfigure[Star]{\includegraphics[width=0.10\textwidth]{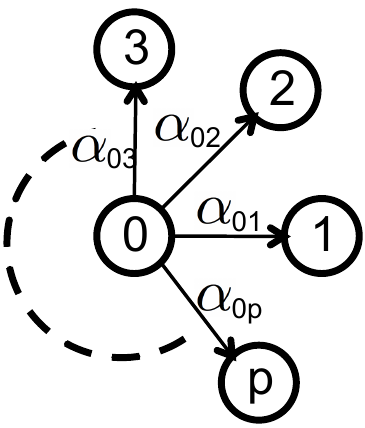} \label{fig:star}}
  \vspace{-4mm}
  \caption{Example networks.} \label{fig:canonical-networks}
  \vspace{-3mm}
\end{figure}

{\bf Remarks for conditions 3 and 4} Well-known pairwise transmission likelihoods such as exponential, Rayleigh or Power-law, used in previous
work~\cite{manuel11icml}, satisfy conditions 3 and 4.

\vspace{-3mm}
\section{Sample Complexity}
\label{sec:results}
\vspace{-2mm}

\emph{How many cascades do we need to recover the network structure?} We will answer this question by providing a sample complexity analysis of the optimization in 
Eq.~\ref{eq:opt-problem-regularized-one-node}. Given the conditions spelled out in Section~\ref{sec:conditions}, we can show that the number of cascades needs to grow 
polynomially in the number of true parents of a node, and depends only lo\-ga\-rith\-mi\-ca\-lly on the size of the network. This is a positive result, since the network size can 
be very large (millions or billions), but the number of parents of a node is usually small compared the network size. More specifically, for each individual node, we have
the following result:
\begin{theorem} \label{th:main-result}
\vspace{-2mm}
Consider an instance of the continuous-time diffusion model with parameters $\alpha^{*}_{ji}$ and associated edges $\Ecal^{*}$ such that
the model satisfies condition 1-4, and let $C^n$ be a set of
$n$ cascades drawn from the model. Suppose that the regularization parameter $\lambda_n$ is selected to satisfy
\begin{equation}
\lambda_n \geq  8 k_3 \frac{2-\varepsilon}{\varepsilon} \sqrt{\frac{\log p}{n}}.
\end{equation}
Then, there exist positive constants $L$ and $K$, independent of $(n, p, d)$, such that if
\begin{equation}
n > L d^3 \log p, 
\end{equation}
then the following properties hold with probability at least $1-2\exp(-K\lambda^2_n n)$:
\begin{enumerate}[noitemsep,nolistsep]
\item For each node $i \in \Vcal$, the $\ell_1$-regularized network in\-fe\-rence problem defined in Eq.~\ref{eq:opt-problem-regularized-one-node} has a unique solution,
and so uniquely specifies a set of incoming edges of node $i$.
\item For each node $i \in \Vcal$, the estimated set of incoming edges does not include any false edges and include all true edges.
\end{enumerate}
Furthermore, suppose that the finite sample Hessian matrix $\mathcal{Q}^{n}$ satisfies conditions 1 and 2. Then there exist positive constants $L$ and $K$, independent of $(n, p, d)$, such that the sample complexity can be improved to $n > L d^2 \log p$ with other statements remain the same.
\vspace{-2mm}
\end{theorem}

{\bf Remarks.} The above sample complexity is proved for each node separately for re\-co\-ve\-ring its parents. Using a union bound, we can provide the sample complexity for re\-co\-ve\-ring the entire network structure by joining these parent-child relations together. The resulting sample complexity and the choice of regularization parameters will remain largely the same, except that the dependency on $d$ will change from $d$ to $d_{max}$ (the largest number of parents of a node), and the dependency on $p$ will change from $\log p$ to $2 \log N$ ($N$ the number of nodes in the network).

\vspace{-3mm}
\subsection{Outline of Analysis}
\vspace{-2mm}
The proof of Theorem~\ref{th:main-result} uses a technique called primal-dual witness method, previously used in the proof of sparsistency of Lasso~\cite{wainwright2009sharp} and high-dimensional Ising model selection~\cite{ravikumar2010high}.
To the best of our knowledge, the present work is the first that uses this technique in the context of diffusion network in\-fe\-rence.
First, we show that the optimal solutions to Eq.~\ref{eq:opt-problem-regularized-one-node} have shared sparsity pattern, and under a further condition, the solution is unique (proven in Appendix~\ref{app:uniqueness}):
\begin{lemma} \label{lemma:uniqueness}
\vspace{-2mm}
Suppose that there exists an optimal primal-dual solution $(\boldsymbol{\hat{\alpha}}, \boldsymbol{\hat{\mu}})$ to Eq.~\ref{eq:opt-problem-regularized-one-node} with an associated subgradient vector $\mathbf{\hat{z}}$ such that $||\mathbf{\hat{z}}_{S^c}||_{\infty} < 1$. Then, any optimal primal solution $\boldsymbol{\tilde{\alpha}}$ 
must have $\boldsymbol{\tilde{\alpha}}_{S^c} = 0$. Moreover, if the Hessian sub-matrix $\Qcal^n_{SS}$ is strictly positive definite, then $\boldsymbol{\hat{\alpha}}$ is the unique optimal
solution.
\vspace{-2mm}
\end{lemma}

Next, we will construct a primal-dual vector $(\boldsymbol{\hat{\alpha}}, \boldsymbol{\hat{\mu}})$ along with an associated subgradient vector $\mathbf{\hat{z}}$.
Furthermore, we will show that, under the
assumptions on $(n, p, d)$ stated in Theorem~\ref{th:main-result}, our constructed solution satisfies the KKT optimality conditions to Eq.~\ref{eq:opt-problem-regularized-one-node}, and the primal vector has the same sparsity pattern as the true parameter $\boldsymbol{\alpha}^*$,~\ie,
\begin{align}
\hat{\alpha}_j > 0,\ \forall j : \alpha^{*}_j > 0, \label{eq:correctness-1} \\
\hat{\alpha}_j = 0,\ \forall j : \alpha^{*}_j = 0. \label{eq:correctness-2}
\end{align}
Then, based on Lemma~\ref{lemma:uniqueness},  we can deduce that the optimal solution to Eq.~\ref{eq:opt-problem-regularized-one-node} correctly recovers the sparsisty pattern of $\boldsymbol{\alpha}^*$, and thus the incoming edges to node $i$.

More specifically, we start by realizing that a primal-dual optimal solution $(\boldsymbol{\tilde{\alpha}},\boldsymbol{\tilde{\mu}})$ to Eq.~\ref{eq:opt-problem-regularized-one-node} must satisfy 
the ge\-ne\-ra\-lized Karush-Kuhn-Tucker (KKT)
conditions~\cite{boyd2004convex}:
\begin{align}
0 \in \nabla \ell^{n}(\boldsymbol{\tilde{\alpha}}) + \lambda_{n} \mathbf{\tilde{z}} - \boldsymbol{\tilde{\mu}}, \label{eq:kkt-1} \\
\tilde{\mu}_j \tilde{\alpha}_j = 0, \label{eq:kkt-2} \\
\tilde{\mu}_j \geq 0, \label{eq:kkt-3} \\
\tilde{z}_j = 1,\ \forall \tilde{\alpha}_j > 0, \label{eq:kkt-4} \\
|\tilde{z}_j| \leq 1,\ \forall \tilde{\alpha}_j = 0, \label{eq:kkt-5}
\end{align}
where $\ell^{n}(\boldsymbol{\tilde{\alpha}}) = - \frac{1}{n} \sum_{c \in C^n} \log g(\casc^c; \boldsymbol{\tilde{\alpha}})$ and $\mathbf{\tilde{z}}$ denotes the sub\-gra\-dient of the $\ell_1$-norm.

Suppose the true set of parent of node $i$ is $S$. We construct the primal-dual vector $(\boldsymbol{\hat{\alpha}}, \boldsymbol{\hat{\mu}})$ and the associated sub\-gra\-dient vector 
$\mathbf{\hat{z}}$ in the following way
\begin{enumerate}[noitemsep,nolistsep]
\item We set $\boldsymbol{\hat{\alpha}}_S$ as the solution to the partial regularized maximum likelihood problem
\begin{equation} \label{eq:restricted-optimization}
\boldsymbol{\hat{\alpha}}_S = \argmin_{(\boldsymbol{\alpha}_S, 0), \boldsymbol{\alpha}_S \geq 0} \ \{ \ell^{n}(\boldsymbol{\alpha}) + \lambda_{n} ||\boldsymbol{\alpha}_S||_1 \}. \\
\end{equation}
Then, we set $\boldsymbol{\hat{\mu}}_S \geq 0$ as the dual solution associated to the primal solution $\boldsymbol{\hat{\alpha}}_S$.
\item We set $\boldsymbol{\hat{\alpha}}_{S^c} = 0$, so that condition~\eqref{eq:correctness-2} holds, and $\boldsymbol{\hat{\mu}}_{S^c} = \boldsymbol{\mu^*}_{S^c} \geq 0$, where
$\boldsymbol{\mu^*}$ is the optimal dual solution to the following problem:
\begin{equation}
	\label{eq:opt-problem-one-node-population}
	\begin{array}{ll}
		\mbox{minimize$_{\boldsymbol{\alpha}}$} & \EE_c\sbr{\ell^{n}(\boldsymbol{\alpha})} \\
		\mbox{subject to} & \alpha_{j} \geq 0,\, j=1,\ldots,N, i \neq j.
	\end{array}
\end{equation}
Thus, our construction satisfies condition~\eqref{eq:kkt-3}.
%
\item We obtain $\mathbf{\hat{z}}_{S^c}$ from~\eqref{eq:kkt-1} by substituting in the constructed $\boldsymbol{\hat{\alpha}}$,  $\boldsymbol{\hat{\mu}}$ and $\mathbf{\hat{z}}_S$.
\end{enumerate}
Then, we only need to prove that, under the stated sca\-lings of $(n, p, d)$, with high-probability, the remaining KKT
conditions~\eqref{eq:correctness-1},~\eqref{eq:kkt-2},~\eqref{eq:kkt-4} and~\eqref{eq:kkt-5} hold.

For simplicity of exposition, we first assume that the dependency and incoherence conditions hold for the finite sample Hessian matrix $\mathcal{Q}^{n}$. Later we will lift this restriction and only place these conditions on the population Hessian matrix $\mathcal{Q}^*$. The following lemma show that our constructed solution satisfies condition~\eqref{eq:correctness-1}:

\begin{lemma} \label{lemma:alphas-s}
\vspace{-2mm}
Under condition 3, if the regularization pa\-ra\-me\-ter is selected to satisfy
\begin{equation*}
\sqrt{d} \lambda_n \leq \frac{C_{\min}^2}{6 (k_2 + 2 k_1 \sqrt{C_{\max}})},
\end{equation*}
and $\| \nabla_s \ell^n(\boldsymbol{\alpha}^*) \|_{\infty} \leq \frac{\lambda_n}{4}$, then,
\begin{equation*}
\| \boldsymbol{\hat{\alpha}}_S - \boldsymbol{\alpha}^{*}_S  \|_2 \leq 3 \sqrt{d} \lambda_n / C_{\min} \leq \alpha^*_{\min}/2,
\end{equation*}
as long as $\alpha^*_{\min} \geq 6 \sqrt{d} \lambda_n / C_{\min}$.
\vspace{-2mm}
\end{lemma}

Based on this lemma, we can then further show that the KKT conditions~\eqref{eq:kkt-2} and~\eqref{eq:kkt-4} also hold for the constructed solution. This can be trivially deduced from condition~\eqref{eq:correctness-1} and ~\eqref{eq:correctness-2}, and our construction steps (a) and (b). Note that it also implies that $\boldsymbol{\hat{\mu}}_S = \boldsymbol{\mu}^{*}_S = 0$, and hence $\boldsymbol{\hat{\mu}} = \boldsymbol{\mu}^{*}$.

Proving condition~\eqref{eq:kkt-5} is more challenging. We first provide more details on how to construct $\mathbf{\hat{z}}_{S^c}$ mentioned in step (c). We start by using a Taylor expansion of Eq.~\ref{eq:kkt-1},
\begin{equation}
\mathcal{Q}^n (\boldsymbol{\hat{\alpha}} - \boldsymbol{\alpha}^{*}) = -\nabla \ell^n(\boldsymbol{\alpha}^*) - \lambda_n \boldsymbol{\hat{z}} + \boldsymbol{\hat{\mu}} - \mathbf{R}^{n},\label{eq:kkt-1-taylor}
\end{equation}
where $\mathbf{R}^n$ is a remainder term with its $j$-th entry
\begin{equation*}
	R^n_j = \big[\nabla^2 \ell^n(\boldsymbol{\bar{\alpha}}_j) - \nabla^2 \ell^n(\boldsymbol{\alpha^*}) \big]_j^T (\boldsymbol{\hat{\alpha}} - \boldsymbol{\alpha^*}),
\end{equation*}
and $\boldsymbol{\bar{\alpha}}_j = \theta_j \boldsymbol{\hat{\alpha}} + (1-\theta_j) \boldsymbol{\alpha^*}$ with $\theta_j \in [0, 1]$ according to the mean value theorem. Rewriting Eq.~\ref{eq:kkt-1-taylor} using block matrices
\begin{multline*}
	\begin{pmatrix}
		\mathcal{Q}_{SS}^n & \mathcal{Q}_{SS^c}^n \\
		\mathcal{Q}_{S^cS}^n & \mathcal{Q}_{S^cS^c}^n
	\end{pmatrix}
	\begin{pmatrix}
	\boldsymbol{\hat{\alpha}}_S - \boldsymbol{\alpha}^*_S  \\
	\boldsymbol{\hat{\alpha}}_{S_{c}} - \boldsymbol{\alpha}^*_{S^c}
	\end{pmatrix} \\
	= - \begin{pmatrix}
	\nabla_S	\ell^n(\boldsymbol{\alpha}^*)  \\
	\nabla_{S^c}	\ell^n(\boldsymbol{\alpha}^*)
	\end{pmatrix}
	- \lambda_n \begin{pmatrix}
	\boldsymbol{\hat{z}}_S  \\
	\boldsymbol{\hat{z}}_{S^c}
	\end{pmatrix}
	+ \begin{pmatrix}
	\mathbf{\hat{\mu}}_S  \\
	\mathbf{\hat{\mu}}_{S^c}
	\end{pmatrix}
	- \begin{pmatrix}
	\mathbf{R}_S^n  \\
	\mathbf{R}_{S^c}^n
	\end{pmatrix}
\end{multline*}
and, after some algebraic manipulation, we have
\begin{multline*}
	\lambda \mathbf{\hat{z}}_{S^c} = - \nabla_{S^c} \ell^n(\boldsymbol{\alpha}^*) + \boldsymbol{\hat{\mu}}_{S^c} - \mathbf{R}_{S^c}^n\\
	- \mathcal{Q}_{S^cS}^n (\mathcal{Q}_{SS}^n)^{-1} \big(- \nabla_{s} \ell^n(\boldsymbol{\alpha}^*) - \lambda \boldsymbol{\hat{z}}_S   + \boldsymbol{\hat{\mu}}_S - \mathbf{R}_S^n\big).
\end{multline*}
Next, we upper bound $\|\mathbf{\hat{z}}_{S^c}\|_{\infty}$ using the triangle inequality
\begin{eqnarray*}
	& \|\mathbf{\hat{z}}_{S^c} \|_{\infty} & \leq  \lambda_n^{-1} \|\boldsymbol{\mu}^{*}_{S^c}- \nabla_{S^c}
	 \ell^n(\boldsymbol{\alpha}^*) \|_{\infty}+ \lambda_n^{-1} \| \mathbf{R}^n_{S^c} \|_{\infty}	\\
	& & + \| \mathcal{Q}_{S^cS}^n (\mathcal{Q}_{SS}^n)^{-1} \|_{\infty} \times \big[ 1 +
	\lambda_n^{-1} \|\mathbf{R}^n_S\|_{\infty} \\
	& & + \lambda_n^{-1} \| \mu^*_S - \nabla_S \ell^n(\boldsymbol{\alpha}^*) \|_{\infty} \big],
\end{eqnarray*}
and we want to prove that this upper bound is smaller than $1$. This can be done with the help of the following two lemmas (proven in Appendices~\ref{app:proof-gradient-hoeffding} and~\ref{app:proof-taylor-error}):
\begin{lemma} \label{lemma:gradient}
\vspace{-2mm}
Given $\varepsilon \in (0, 1]$ from the incoherence condition, we have,
\begin{multline*}
	P\left(\frac{2-\varepsilon}{\lambda_n}\| \nabla \ell^n(\boldsymbol{\alpha}^*) - \boldsymbol{\mu}^*
	\|_{\infty} \geq 4^{-1} \varepsilon\right)
	\\ \leq 2 p \exp(- \frac{n \lambda^2_n \varepsilon^2}{32 k^2_3 \left(2-\varepsilon\right)^2}),
\end{multline*}
which converges to zero at rate $\exp(-c \lambda_n^2 n)$ as long as $\lambda_n \geq 8 k_3 \frac{2-\varepsilon}{\varepsilon} \sqrt{\frac{\log p}{n}}$.
\vspace{-1mm}
\end{lemma}
\begin{lemma} \label{lemma:taylor-error}
\vspace{-1mm}
Given $\varepsilon \in (0, 1]$ from the incoherence condition, if conditions 3 and 4 holds, $\lambda_n$ is selected to satisfy
\begin{equation*}
\lambda_n d \leq C_{\min}^{2} \frac{\varepsilon}{36 K (2-\varepsilon)},
\end{equation*}
where $K =  k_1 + k_4 k_1 + k_1^2 + k_1 \sqrt{C_{\max}}$, and $\| \nabla_s \ell^n(\boldsymbol{\alpha}^*) \|_{\infty} \leq \frac{\lambda_n}{4}$, then,
$
\frac{\|\mathbf{R}^n\|_{\infty}}{\lambda_n} \leq \frac{\varepsilon}{4(2-\varepsilon)},
$
as long as $\alpha^*_{\min} \geq 6 \sqrt{d} \lambda_n / C_{\min}$.
\vspace{-2mm}
\end{lemma}
Now, applying both lemmas and the incoherence condition on the finite sample Hessian matrix $\mathcal{Q}^{n}$, we have
\begin{eqnarray*}
	& \|\mathbf{\hat{z}}_{S^c} \|_{\infty} & \leq~~  (1-\varepsilon) + \lambda_n^{-1} (2-\varepsilon) \| \mathbf{R}^n \|_{\infty} \\
	& &~~ + \lambda_n^{-1} (2-\varepsilon) \| \boldsymbol{\mu}^* - \nabla \ell^n(\boldsymbol{\alpha}^*) \|_{\infty} \\
	& & \leq ~~(1-\varepsilon) + 0.25 \varepsilon + 0.25 \varepsilon =  1- 0.5 \varepsilon, 
\end{eqnarray*}
and thus condition~\eqref{eq:kkt-5} holds.

A possible choice of the regularization parameter $\lambda_n$ and cascade set size $n$ such that the conditions of the Lemmas~\ref{lemma:alphas-s}-\ref{lemma:taylor-error} are sa\-tis\-fied is
$\lambda_n = 8 k_3 (2-\varepsilon) \varepsilon^{-1} \sqrt{n^{-1} \log p}$ and $n > 288^2k_3^2 (2-\varepsilon)^4 C_{min}^{-4} \varepsilon^{-4} d^2 \log p + \left(48 k_3 (2-\varepsilon) C_{min}^{-1} (\alpha_{min}^*)^{-1} \varepsilon^{-1} \right)^2 d \log p$.

Last, we lift the dependency and incoherence conditions imposed on the finite sample Hessian matrix $\mathcal{Q}^n$. We show that if we only impose these conditions in the corres\-pon\-ding population matrix $\mathcal{Q}^*$, then they will also hold for $\mathcal{Q}^n$ with high probability (proven in Appendices~\ref{app:dependency-n} and~\ref{app:coherence-n}).
\begin{lemma} \label{lemma:dependency-n}
\vspace{-2mm}
If condition 1 holds for $\mathcal{Q}^*$, then, for any $\delta > 0$,
\begin{align*}
	P\left(\Lambda_{min}\left(\mathcal{Q}_{SS}^n\right) \leq C_{min} - \delta\right) \leq 2 d^{B_1} \exp(-A_1 \frac{\delta^2 n}{d^2}),
\end{align*}
\begin{align*}
	P\left(\Lambda_{max}\left(\mathcal{Q}_{SS}^n\right) \geq C_{max} + \delta\right) \leq 2 d^{B_2} \exp(-A_2 \frac{\delta^2 n}{d^2}),
\end{align*}
%
%
where $A_1$, $A_2$, $B_1$ and $B_2$ are constants independent of $(n, p, d)$.
\vspace{-1mm}
\end{lemma}
\begin{lemma} \label{lemma:coherence-n}
\vspace{-1mm}
If $||| \mathcal{Q}^{*}_{S^c S} \left(\mathcal{Q}^{*}_{SS}\right)^{-1} |||_{\infty} \leq 1 - \varepsilon$, then,
\begin{align*}
	P\left(\| \mathcal{Q}_{S^cS}^n (\mathcal{Q}_{SS}^n)^{-1} \|_{\infty} \geq 1 - \varepsilon/2\right)
	\leq p\exp(-K \frac{n}{d^3}),
\end{align*}
where $K$ is a cons\-tant independent of $(n, p, d)$.
\vspace{-2mm}
\end{lemma}

Note in this case the cascade set size need to increase to $n >  L d^3 \log p$, where $L$ is a sufficiently large positive constant independent of
$(n, p, d)$, for the error probabilities on these last two lemmas to converge to zero.

%% file: 050algorithm.tex
\emph{Can we design efficient algorithms to solve Eq.~\eqref{eq:opt-problem-regularized-one-node} for network recovery?}
Here, we will design a proximal gradient algorithm which is well suited for solving non-smooth, constrained, large-scale or high-dimensional convex optimization
problems~\cite{parikh2013proximal}. Moreover, they are easy to understand, derive, and implement.
\begin{algorithm}[t]
\caption{$\ell_1$-regularized network inference} \label{alg:proximal-gradient-regularized-netrate} 
\begin{algorithmic}
\REQUIRE $C^n, \lambda_n, K, L$
\FORALL{$i \in \Vcal$}
\STATE $k = 0$
\WHILE{$k < K$}
  \STATE $\boldsymbol{\alpha_{i}}^{k+1} = \left( \boldsymbol{\alpha_{i}}^{k} - L \nabla_{\boldsymbol{\alpha}_{i}} \ell^{n}(\boldsymbol{\alpha}_{i}^{k}) - \lambda_n L \right)_{+}$
  \STATE $k = k+1$
  \ENDWHILE
\STATE $\boldsymbol{\hat{\alpha}_{i}} = \boldsymbol{\alpha_{i}}^{K-1}$
\ENDFOR
\RETURN $\{\boldsymbol{\hat{\alpha}_{i}}\}_{i \in \Vcal}$
\end{algorithmic}
\vspace{-1mm}
\end{algorithm}
%
%
We first rewrite Eq.~\ref{eq:opt-problem-regularized-one-node} as an unconstrained optimization problem:
\begin{equation*}
	\label{eq:opt-problem-regularized-one-node-proximal-gradient}
	\begin{array}{ll}
		\mbox{minimize$_{\boldsymbol{\alpha}}$} & \ell^{n}(\boldsymbol{\alpha}) + g(\boldsymbol{\alpha}),
	\end{array}
\end{equation*}
where the non-smooth convex function $g(\boldsymbol{\alpha}) = \lambda_n ||\boldsymbol{\alpha}||_1$ if $\boldsymbol{\alpha} \geq 0$ and $+\infty$ otherwise.
%
Here, the general recipe from~\citet{parikh2013proximal} for designing proximal gra\-dient algorithm can be applied directly.

Algorithm~\ref{alg:proximal-gradient-regularized-netrate} summarizes the resulting algorithm.
In each ite\-ra\-tion of the algorithm, we need to compute $\nabla \ell^{n}$ (Table~\ref{tab:functions}) and the proximal operator
$\mathbf{prox}_{L^k g}(v)$, where $L^{k}$ is a step size that we can set to a constant value $L$ or find using a simple line search~\cite{beck2009gradient}.
Using Moreau'{}s decomposition and the conjugate function $g^{*}$, it is easy to show that the proximal operator for our particular function
$g(\cdot)$ is a soft-thresholding operator, $(v - \lambda_n L^{k})_{+}$, which leads to a sparse optimal solution $\boldsymbol{\hat{\alpha}}$, as desired.

%% file: 060experiments.tex
\begin{figure}[t]
	\centering
	\vspace{-3mm}
	\subfigure[Super-neighborhood $p_i$]{\makebox[4.2cm][c]{\includegraphics[width=0.23\textwidth]{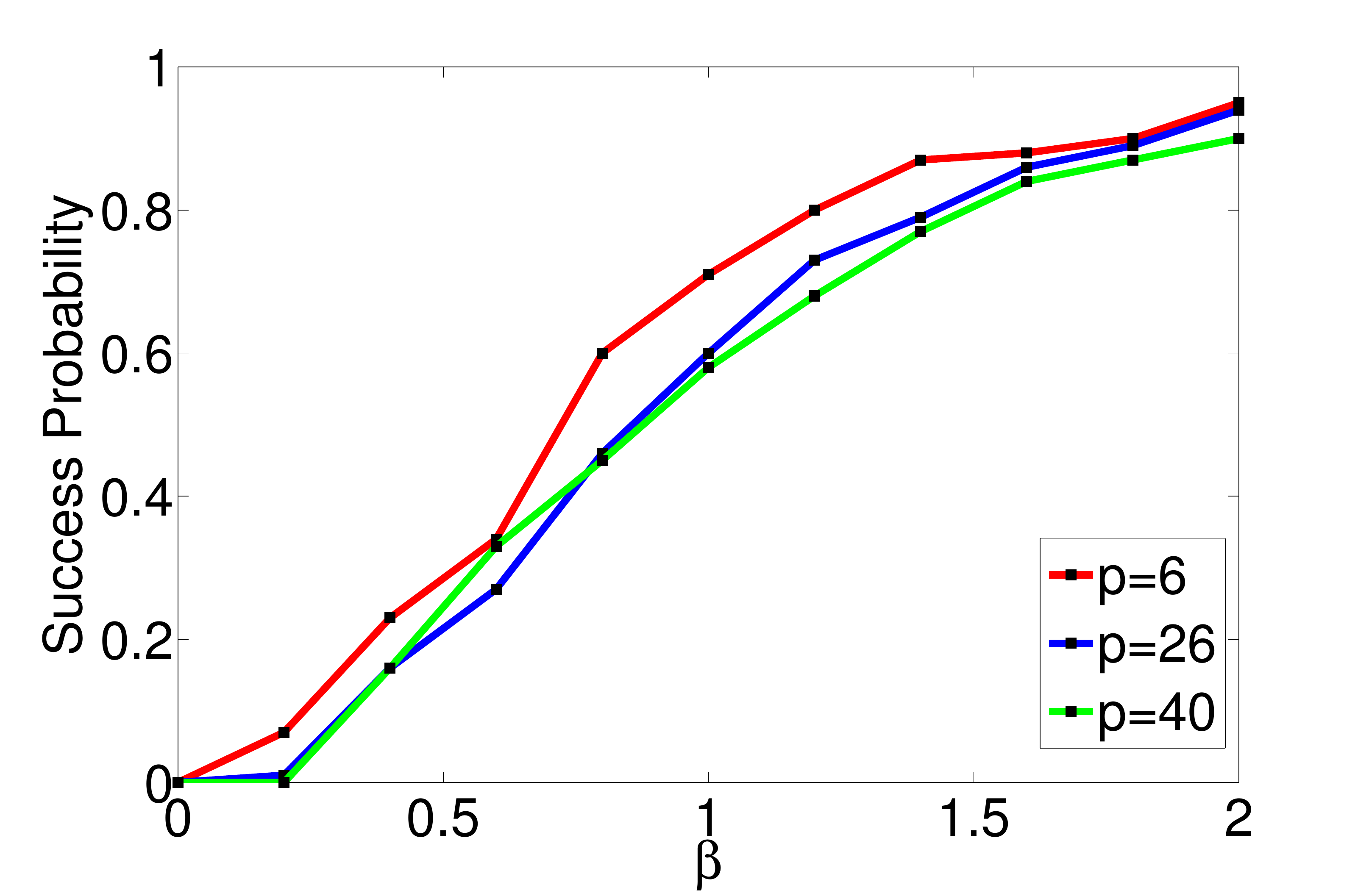}\label{fig:super-neighborhood}}}
	\subfigure[$\lambda_n$]{\includegraphics[width=0.23\textwidth]{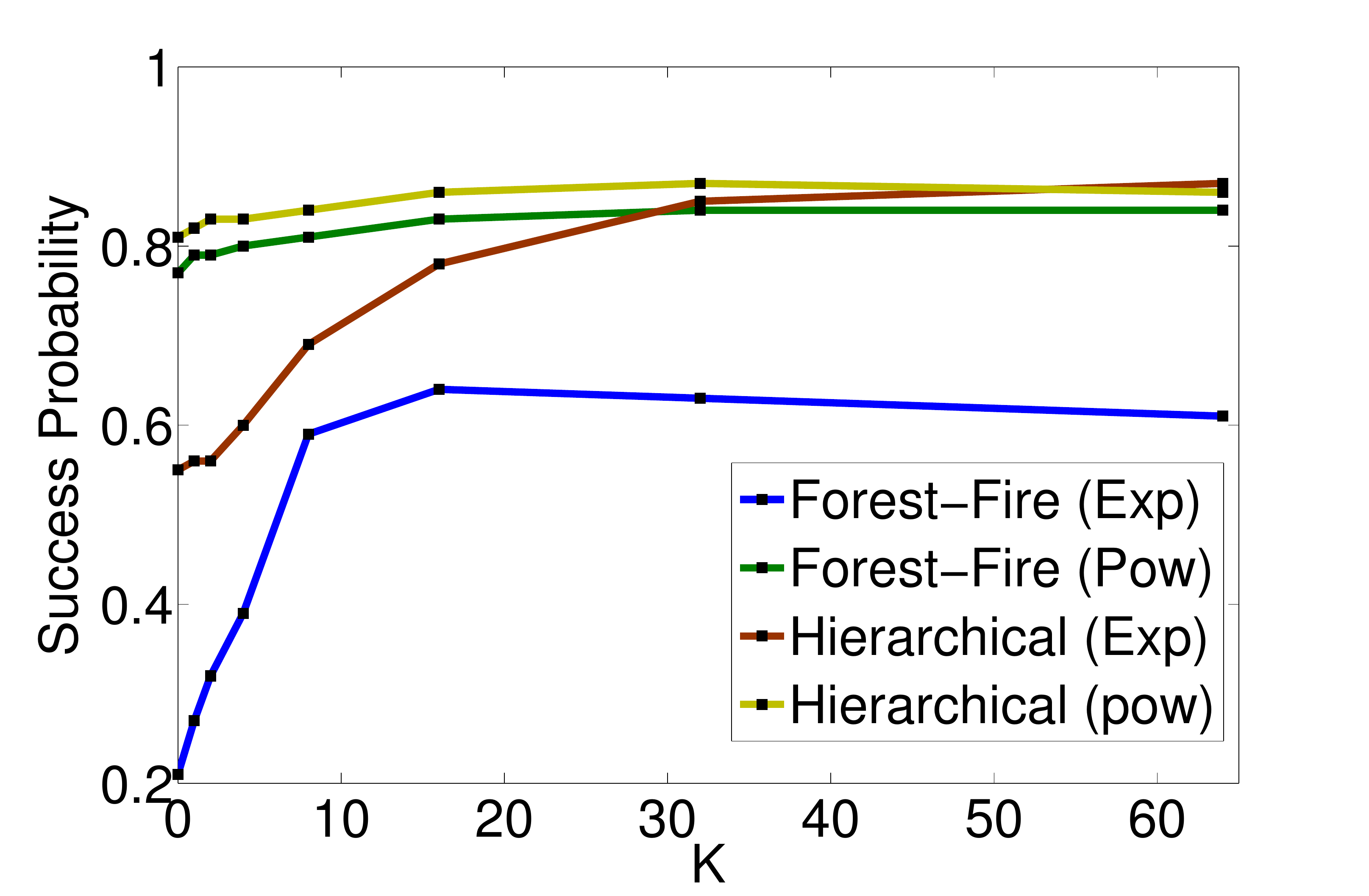}\label{fig:lambda-n}}
	\vspace{-4mm}
 	\caption{Success probability vs. \# of cascades.} \label{fig:consequences}
	\vspace{-3mm}
\end{figure}
%
%
In this section, we first illustrate some consequences of Th.~\ref{th:main-result} by applying our algorithm to several types of networks, parameters $(n, p, d)$, and
regularization parameter $\lambda_n$. Then, we compare our algorithm to two different state-of-the-art algorithms: \netrate~\cite{manuel11icml} and
First-Edge~\cite{abrahaotrace}.

{\bf Experimental Setup} We focus on synthetic networks that mimic the structure of real-world diffusion networks -- in particular, social networks. We consider two models of directed
real-world social networks: the Forest Fire model~\citep{barabasi99emergence} and the Kronecker Graph model~\citep{leskovec2010kronecker}, and use simple pairwise
transmission models such as exponential, power-law or Rayleigh. We use networks with $128$ nodes and, for each edge, we draw its asso\-cia\-ted
transmission rate from a uniform distribution $U(0.5, 1.5)$.
%
We proceed as follows: we generate a network $\Gcal^{*}$ and transmission rates $\alphs^{*}$, simulate a set of cascades and, for each cascade, record the node infection times.
Then, given the infection times, we infer a network $\hat{\Gcal}$.
Finally, when we illustrate the consequences of Th.~\ref{th:main-result}, we evaluate the accuracy of the inferred neighborhood of a node $\hat{\Ncal}^{-}(i)$ using probability of
success $P(\hat{\Ecal} = \Ecal^*)$, estimated by running our method of $100$ independent cascade sets. When we compare our algorithm to \netrate and First-Edge, we use the 
$F_{1}$ score, which is defined as $2 P R / (P+R)$, where precision (P) is the fraction of edges in the inferred network $\hat{\Gcal}$ present in the true network $\Gcal^{*}$, 
and recall (R) is the fraction of edges of the true network $\Gcal^{*}$ present in the inferred network $\hat{\Gcal}$.

{\bf Parameters $(n, p, d)$} According to Th.~\ref{th:main-result}, the number of cascades that are necessary to successfully infer the incoming edges of a node will increase polynomially
to the node'{}s neighborhood size $d_i$ and logarithmically to the super-neighborhood size $p_i$. Here, we infer the incoming links of nodes of a hierarchical
Kronecker network with the same in-degree ($d_i=3$) but different super-neighboorhod set sizes $p_i$ under different scalings $\beta$ of the number of cascades $n = 10 \beta d \log p$ and choose
the regularization pa\-ra\-me\-ter $\lambda_n$ as a constant factor of $\sqrt{\log (p)/ n}$ as suggested by Th.~\ref{th:main-result}.
We used an exponential transmission model and $T = 5$.
Fig.~\ref{fig:super-neighborhood}  summarizes the results, where, for each node, we used cascades which contained at least one node in the super-neighborhood of the node under study. As
predicted by Th.~\ref{th:main-result}, very different $p$ values lead to curves that line up with each other quite well.
%

{\bf Regularization parameter $\lambda_n$} Our main result indicates that the regularization parameter $\lambda_n$ should be a constant factor of $\sqrt{\log (p)/ n}$.
Fig.~\ref{fig:lambda-n} shows the success probability of our algorithm against different scalings $K$ of the regularization parameter $\lambda_n = K \sqrt{\log (p) / n}$ for different
types of networks using $150$ cascades and $T = 5$.
We find that for sufficiently large $\lambda_n$, the success probability flattens, as expected from Th.~\ref{th:main-result}.
It flattens at values smaller than one because we used a fixed number of cascades $n$, which may not satisfy the conditions of
Th.~\ref{th:main-result}.

{\bf Comparison with \netrate and First-Edge} Fig.~\ref{fig:comparison} compares the accuracy of our algorithm, \netrate and First-Edge against number of cascades for a hierarchical
Kronecker network with power-law transmission model and a Forest Fire network with exponential transmission model, with an observation window $T = 10$.
Our method outperforms both competitive methods, finding especially striking the competitive advantage with respect to First-Edge.
%
%
\begin{figure}[t]
	\centering
	\vspace{-3mm}
	\subfigure[Kronecker hierarchical, \pow]{\makebox[4.2cm][c]{\includegraphics[width=0.23\textwidth]{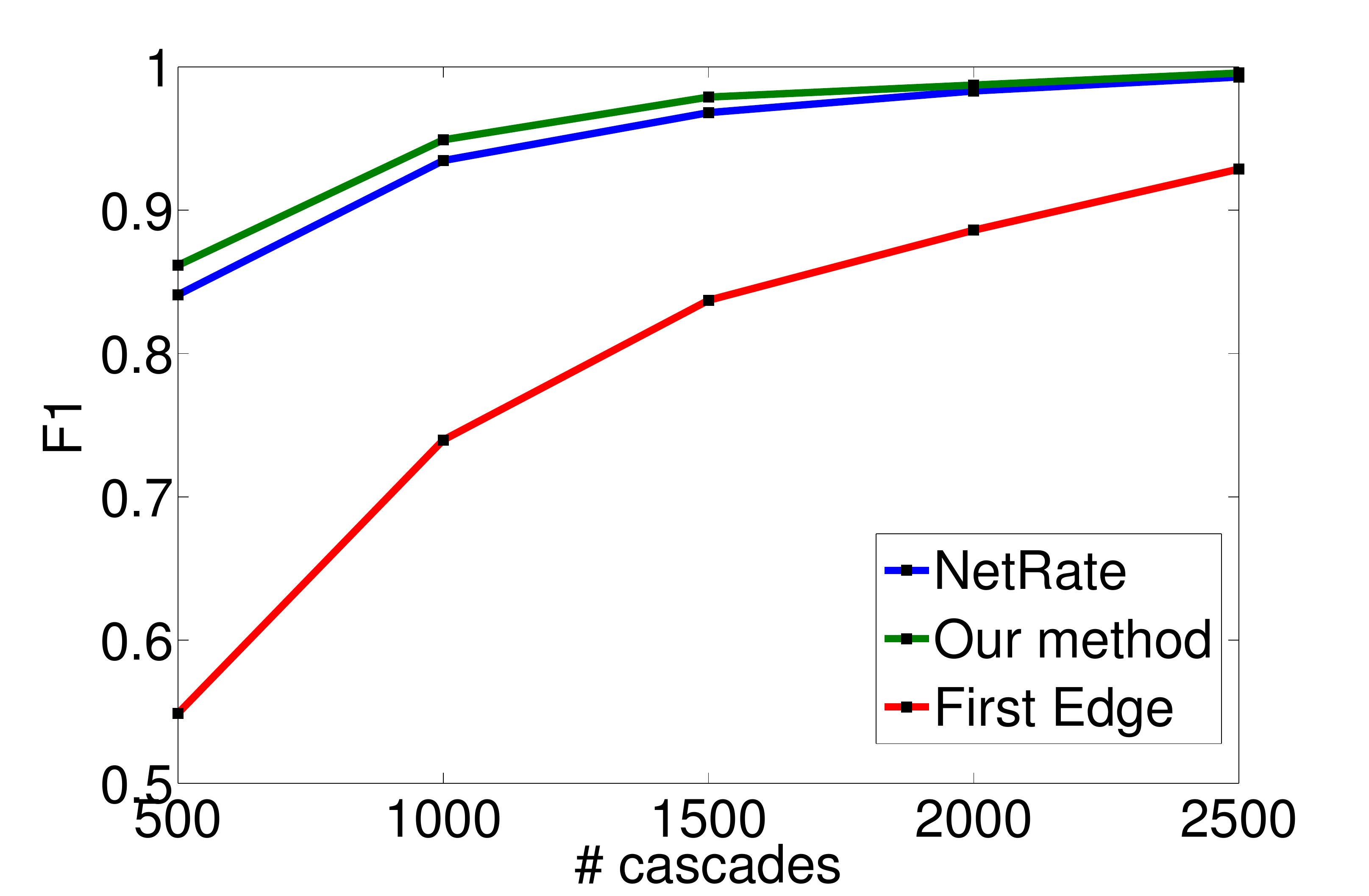}\label{fig:hierarchical-pow}}}
	\subfigure[Forest Fire, \expo]{\includegraphics[width=0.23\textwidth]{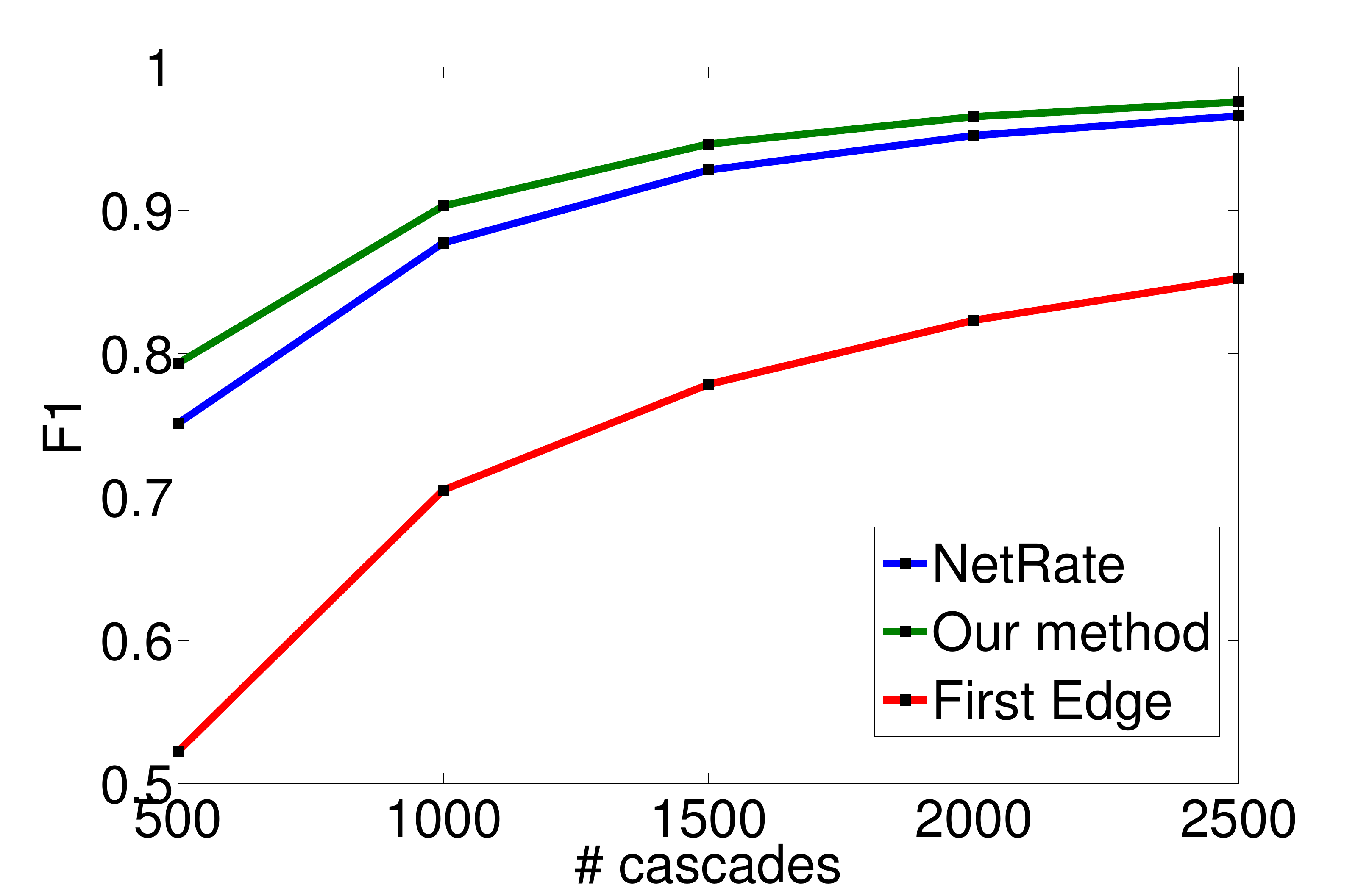}\label{fig:ff-exp}}
	\vspace{-4mm}
 	\caption{$F_{1}$-score vs. \# of cascades.} \label{fig:comparison}
	\vspace{-3mm}
\end{figure}

%% file: 070conclusions.tex
Our work contributes towards establishing a theoretical foundation of the network inference problem.
Specifically, we proposed a $\ell_1$-regularized maximum likelihood in\-fe\-rence method for a well-known continuous-time diffusion model and an efficient proximal gradient implementation, 
and then show that, for general networks satisfying a natural incoherence condition, our method achieves an exponentially decreasing error  with respect to the number of cascades 
as long as $O(d^3 \log N)$ cascades are recorded.

Our work also opens many interesting venues for future work.
For example, given a fixed number of cascades, it would be useful to provide confidence 
intervals on the inferred edges.
%
%
Further, given a network with arbitrary pairwise likelihoods, it is an open question whether there 
always exists at least one source distribution and time window value such that the incoherence 
condition is satisfied, and, and if so, whether there is an efficient way of finding this distribution.
Finally, our work assumes all activations occur due to network diffusion and are recorded.
It would be interesting to allow for missing observations, as well as activations due to exogenous factors.

%% file: 080appendix.tex
\section{Proof of Lemma~\ref{lem:hessian-positive-definite}} \label{app:proof-hessian-positive-definite}

\begin{lemma} \label{lem:hessian-positive-definite}
Given log-concave survival functions and concave hazard functions in the parameter(s) of the pairwise transmission likelihoods, then, a sufficient condition for the Hessian matrix $\mathcal{Q}^n$ to
be positive definite is that the hazard matrix $X^n(\boldsymbol{\alpha})$ is non-singular.
\end{lemma}
\begin{proof}
Using Eq.~\ref{eq:consistency-hessian-matrix-factorization}, the Hessian matrix can be expressed as a sum of two matrices, $\mathbf{D}^n(\boldsymbol{\alpha})$
and $\mathbf{X}^n(\boldsymbol{\alpha}) \mathbf{X}^n(\boldsymbol{\alpha})^\top $.
The matrix $\mathbf{D}^n(\boldsymbol{\alpha})$ is trivially positive semidefinite by log-concavity of the survival functions and concavity of the hazard functions.
The matrix $\mathbf{X}^n(\boldsymbol{\alpha}) \mathbf{X}^n(\boldsymbol{\alpha})^\top $ is positive definite matrix since $\mathbf{X}^n(\boldsymbol{\alpha})$ is
full rank by assumption.
Then, the Hessian matrix is positive definite since it is a sum a positive semidefinite matrix and a positive definite matrix.
\end{proof}

\section{Proof of Lemma~\ref{lem:nonsingularity}} \label{app:proof-nonsingularity}

\begin{lemma} \label{lem:nonsingularity}
If the source probability $\PP(s)$ is strictly positive for all $s \in \Rcal$, then, for an arbitrarily large number of cascades $n \rightarrow \infty$, there exists an ordering of the nodes and cascades within 
the cascade set such that the hazard matrix $\boldsymbol{X}^n(\boldsymbol{\alpha})$ is non-singular.
\end{lemma}
\begin{proof}
%
In this proof, we find a labeling of the nodes (row indices in $\mathbf{X}^n(\boldsymbol{\alpha})$) and ordering of the cascades (column indices in $\mathbf{X}^n(\boldsymbol{\alpha})$),
such that, for an arbitrary large number of cascades, we can express the matrix $\mathbf{X}^n(\boldsymbol{\alpha})$ as $\left[ T \, B \right]$, where $T \in \mathbb{R}^{p \times p}$ is an
upper triangular with nonzero diagonal elements and $B \in \mathbb{R}^{p \times n-p}$.
And, therefore, $\mathbf{X}^n(\boldsymbol{\alpha})$ has full rank (rank $p$). We proceed first by sorting nodes in $\Rcal$ and then continue by sorting nodes in $\Ucal$:
\begin{itemize}
\denselistA
\item {\bf Nodes in $\Rcal$}: 
For each node $u \in \Rcal$, consider the set of cascades $C_u$ in which $u$ was a source and $i$ got infected. Then, rank each node $u$ according to the earliest position in which 
node $i$ got infected across all cascades in $C_u$ in decreasing order, breaking ties at random.
For example, if a node $u$ was, at least once, the source of a cascade in which node $i$ got infected just after the source, but in contrast, node $v$ was never the source of a cascade
in which node $i$ got infected the second, then node $u$ will have a lower index than node $v$.
Then, assign row $k$ in the matrix $\mathbf{X}^n(\boldsymbol{\alpha})$ to node in position $k$ and assign the first $d$ columns to the corresponding cascades in which
node $i$ got infected earlier.
In such ordering, $\mathbf{X}^n(\boldsymbol{\alpha})_{mk} = 0$ for all $m<k$ and $\mathbf{X}^n(\boldsymbol{\alpha})_{kk} \neq 0$.

\item {\bf Nodes in $\Ucal$}: 
Similarly as in the first step, and assign them the rows $d+1$ to $p$. Moreover, we assign the columns $d+1$ to $p$ to the corresponding cascades in which node 
$i$ got infected earlier. Again, this ordering satisfies that $\mathbf{X}^n(\boldsymbol{\alpha})_{mk} = 0$ for all $m<k$ and $\mathbf{X}^n(\boldsymbol{\alpha})_{kk} \neq 0$.
Finally, the remaining columns $n-p$ can be assigned to the remaining cascades at random.
\vspace{-5mm}
\end{itemize}
This ordering leads to the desired structure $\left[T \, B \right]$, and thus it is non-singular.
\end{proof}

\section{Proof of Eq~\ref{eq:bounded_hazard}.}
\label{app:proof-bounded_hazard}

If the Hazard vector $\boldsymbol{X}(\casc^c ; \boldsymbol{\alpha})$ is Lipschitz continuous in the domain $\{ \boldsymbol{\alpha} : \boldsymbol{\alpha}_S \geq \frac{\alpha^*_{\min}}{2} \}$,
\begin{equation*}
	\|\boldsymbol{X}(\casc^c ; \boldsymbol{\beta}) - \boldsymbol{X}( \casc^c ; \boldsymbol{\alpha})\|_2 \leq k_1 \|\boldsymbol{\beta} - \boldsymbol{\alpha}\|_2,
\end{equation*}
where $k_1$ is some positive constant. Then, we can bound the spectral norm of the difference, $\frac{1}{\sqrt{n}} (\boldsymbol{X}^n(\boldsymbol{\beta}) - \boldsymbol{X}^n(\boldsymbol{\alpha}))$,
in the domain $\{ \boldsymbol{\alpha} : \boldsymbol{\alpha}_S \geq \frac{\alpha^*_{\min}}{2} \}$ as follows:
\begin{align*}
& |\|\frac{1}{\sqrt{n}}
\big(\mathbf{X}^n(\boldsymbol{\beta})
-\mathbf{X}^n(\boldsymbol{\alpha})\big)\| |_2 \\
& = \max_{\|\boldsymbol{u}\|_2=1 } \frac{1}{\sqrt{n}} \| \boldsymbol{u}
\big(\mathbf{X}^n(\boldsymbol{\beta})
-\mathbf{X}^n(\boldsymbol{\alpha})\big) \|_2 \\
& = \max_{\|\boldsymbol{u}\|_2=1 } \frac{1}{\sqrt{n}}
\sqrt{ \sum_{c=1}^n \langle \boldsymbol{u},
\mathbf{X}(\casc^c ; \boldsymbol{\beta}) -
\mathbf{X}(\casc^c ; \boldsymbol{\alpha})\rangle^2 } \\
& \leq \frac{1}{\sqrt{n}} \sqrt{ k_1^2 n \| \boldsymbol{u}
\|_2^2\|\boldsymbol{\beta} -
	\boldsymbol{\alpha}\|_2^2}\\
& \leq 	k_1 \|\boldsymbol{\beta} -\boldsymbol{\alpha}\|_2.
\end{align*}

\section{Proof of Lemma~\ref{lemma:uniqueness}} \label{app:uniqueness}
By Lagrangian duality, the regularized network inference problem defined in Eq.~\ref{eq:opt-problem-regularized-one-node} is equivalent to the following
constrained optimization problem:
\begin{equation}
	\label{eq:opt-problem-regularized-one-node-equivalent}
	\begin{array}{ll}
		\mbox{minimize$_{\boldsymbol{\alpha}_i}$} & \ell^{n}(\boldsymbol{\alpha}_i) \\
		\mbox{subject to} & \alpha_{ji} \geq 0,\, j=1,\ldots,N, i \neq j, \\
		& ||\boldsymbol{\alpha}_{i}||_1 \leq C(\lambda_n)
	\end{array}
\end{equation}
where $C(\lambda_n) < \infty$ is a positive constant. In this alternative formulation, $\lambda_n$ is the Lagrange multiplier for the second constraint. Since $\lambda_n$ is strictly
positive, the constraint is active at any optimal solution, and thus $||\boldsymbol{\alpha}_i||_{1}$ is constant across all optimal solutions.

Using that $\ell^{n}(\boldsymbol{\alpha}_i)$ is a differentiable convex function by assum\-ption and $\{ \boldsymbol{\alpha} : \alpha_{ji} \geq 0, ||\boldsymbol{\alpha}_{i}||_1 \leq C(\lambda_n) \}$
is a convex set, we have that $\nabla \ell^{n}(\boldsymbol{\alpha}_i)$ is constant across optimal primal solutions~\cite{mangasarian1988simple}.
Moreover, any optimal primal-dual solution in the original problem must satisfy the KKT conditions in the alternative formulation defined by Eq.~\ref{eq:opt-problem-regularized-one-node-equivalent},
in particular,
\begin{equation*}
\nabla \ell^{n}(\boldsymbol{\alpha}_i) = -\lambda_n \mathbf{z} + \boldsymbol{\mu},
\end{equation*}
where $\boldsymbol{\mu} \geq 0$ are the Lagrange multipliers associated to the non negativity constraints and $\mathbf{z}$ denotes the subgradient of the $\ell$1-norm.

Consider the solution $\boldsymbol{\hat{\alpha}}$ such that $||\mathbf{\hat{z}_{S^c}}||_{\infty} < 1$ and thus $\nabla_{{\alpha}_{S^c}} \ell^{n}(\boldsymbol{\hat{\alpha}}_i) = -\lambda_n
\mathbf{\hat{z}_{S^c}} + \boldsymbol{\hat{\mu}}_{S^c}$.
Now, assume there is an optimal primal solution $\boldsymbol{\tilde{\alpha}}$ such that $\tilde{\alpha}_{ji} > 0$ for some $j \in S^{c}$, then, using that the gradient must be constant across
optimal solutions, it should hold that $-\lambda_n \hat{z}_j + \hat{\mu}_j = -\lambda_n$, where $\tilde{\mu}_{ji} = 0$ by complementary slackness, which implies
$\hat{\mu}_j = -\lambda_n (1-\hat{z}_j) < 0$. Since $\hat{\mu}_j \geq 0$ by assumption, this leads to a contradiction.
Then, any primal solution $\boldsymbol{\tilde{\alpha}}$ must satisfy $\boldsymbol{\tilde{\alpha}}_{S^c} = 0$ for the gradient to be constant across optimal solutions.

Finally, since $\boldsymbol{\alpha}_{S^c} = 0$ for all optimal solutions, we can consider the restricted optimization problem defined in Eq.~\ref{eq:restricted-optimization}.
If the Hessian sub-matrix $[\nabla^2 L(\boldsymbol{\hat{\alpha}})]_{SS}$ is strictly positive definite, then this restricted optimization problem is strictly convex and the optimal
solution must be unique.

\section{Proof of Lemma~\ref{lemma:alphas-s}} \label{app:proof-balltrick}

To prove this lemma, we will first construct a function
\begin{multline*}
	G(\mathbf{u}_S) := \ell^n(\boldsymbol{\alpha^*}_S+\mathbf{u}_S) -
	\ell^n(\boldsymbol{\alpha^*}_S)\\
	 + \lambda_n	(\|\boldsymbol{\alpha^*}_S+\mathbf{u}_S\|_1 -
	 \|\boldsymbol{\alpha^*}_S\|_1).
\end{multline*}
whose domain is restricted to the convex set $\Ucal = \{ \mathbf{u}_S : \boldsymbol{\alpha}^*_S + \mathbf{u}_S \geq \mathbf{0} \}$. By construction, $G(\mathbf{u}_S)$ has the following properties
\begin{enumerate}[noitemsep,nolistsep]
  \item It is convex with respect to $\mathbf{u}_S$.
  \item Its minimum is obtained at $\boldsymbol{\hat{u}}_S  := \boldsymbol{\hat{\alpha}}_S - \boldsymbol{\alpha}^*_S$. That is $G(\boldsymbol{\hat{u}}_S) \leq G(\mathbf{u}_S)$, $\forall \mathbf{u}_S \neq \boldsymbol{\hat{u}}_S$.  
  \item $G(\boldsymbol{\hat{u}}_S) \leq G(\mathbf{0}) = 0$. 
\end{enumerate}
Based on property 1 and 3, we deduce that any point in the segment,  $\mathbb{L}:=\cbr{\boldsymbol{\tilde{u}}_S: \boldsymbol{\tilde{u}}_S=t \boldsymbol{\hat{u}}_S + (1-t)\mathbf{0}, t\in [0,1]}$, connecting $\boldsymbol{\hat{u}}_S$ and $\mathbf{0}$ has $G(\boldsymbol{\tilde{u}}_S)\leq 0$. That is 
\begin{align*}
    G(\boldsymbol{\tilde{u}}_S) 
    &= G(t \boldsymbol{\hat{u}}_S + (1-t)\mathbf{0})\\
    &\leq t G(\boldsymbol{\hat{u}}_S) + (1-t) G(\mathbf{0}) \leq 0. 
\end{align*}

Next, we will find a sphere centered at $\mathbf{0}$ with strictly po\-si\-tive radius $B$, $\mathbb{S}(B):=\cbr{\boldsymbol{u}_S: \nbr{\boldsymbol{u}_S}_2=B}$, such that function $G(\boldsymbol{u}_S)>0$ (strictly positive) on $\mathbb{S}(B)$. We note that this sphere $\mathbb{S}(B)$ can not intersect with the seg\-ment $\mathbb{L}$ since the two sets have strictly different function values. Furthermore, the only possible configuration is that the segment is contained inside the sphere entirely, leading us to conclude that the end point $\boldsymbol{\hat{u}}_S  := \boldsymbol{\hat{\alpha}}_S - \boldsymbol{\alpha}^*_S$ is also within the sphere. That is $\nbr{\boldsymbol{\hat{\alpha}}_S - \boldsymbol{\alpha}^*_S}_2 \leq B$. 

In the following, we will provide details on finding such a suitable $B$ which will be a function of the regularization parameter $\lambda_n$ and the neighborhood size $d$.
More spe\-ci\-fi\-ca\-lly, we will start by applying a Taylor series expansion and the mean value theorem,
\begin{multline} \label{equ:lemma6-main-equation-for-G}
G( \mathbf{u}_S) = \nabla_S \ell^n(\boldsymbol{\alpha}^*_S)^\top   \mathbf{u}_S \\
+  \mathbf{u}_S^\top  \nabla^2_{SS} \ell^n(\boldsymbol{\alpha}^*_S+b
\mathbf{u}_S )  \mathbf{u}_S\\
	 +\lambda_n(\|\boldsymbol{\alpha}^*_S +  \mathbf{u}_S\|_1 - \|\boldsymbol{\alpha}^*_S\|_1),
\end{multline}
where $b \in [0, 1]$. We will show that $G(\mathbf{u}_S)>0$ by boun\-ding below each term of above
equation separately. 

We bound the absolute value of the first term using the assum\-ption on the gradient, $\nabla_S\ell(\cdot)$,
\begin{multline} \label{equ:lemma6-bound-first-term}
 | \nabla_S\ell^n(\boldsymbol{\alpha}^*_S)^\top  \mathbf{u}_S | \leq \| \nabla_S
 \ell \|_{\infty} \|\mathbf{u}_S\|_1  \\
 \leq \| \nabla_S \ell \|_{\infty} \sqrt{d} \|\mathbf{u}_S \|_2 \\
 \leq 4^{-1}\lambda_n B \sqrt{d}.
\end{multline}

We bound the absolute value of the last term using the reverse triangle inequality.
\begin{equation} \label{equ:lemma6-bound-third-term}
 \lambda_n | \| \boldsymbol{\alpha}^*_S + \mathbf{u}_S \|_1 - \|
 \boldsymbol{\alpha}^*_S \|_1 | \leq \lambda_n \| \mathbf{u}_S \|_1 \leq
 \lambda_n \sqrt{d} \| \mathbf{u}_S \|_2.
\end{equation}

Bounding the remaining middle term is more challenging. We start by rewriting the Hessian as a sum of two matrices, using
Eq.~\ref{eq:consistency-hessian-matrix-factorization},
\begin{align*}
	q &= \min_{\mathbf{u}_S} \mathbf{u}_S^\top
	\mathbf{D}^n_{SS}(\boldsymbol{\alpha}^*_S+b \mathbf{u}_S) \mathbf{u}_S \\
	&+ n^{-1} \mathbf{u}_S^\top  \mathbf{X}^n_S(\boldsymbol{\alpha}^*_S+b \mathbf{u}_S)
	\mathbf{X}^n_S(\boldsymbol{\alpha}^*_S+b \mathbf{u}_S)^\top \mathbf{u}_S \\
	&= \min_{\mathbf{u}_S} \mathbf{u}_S^\top
	\mathbf{D}^n_{SS}(\boldsymbol{\alpha}^*_S+b \mathbf{u}_S) \mathbf{u}_S
	+ \| \mathbf{u}_S^\top  \mathbf{X}^n_S(\boldsymbol{\alpha}^*_S+b \mathbf{u}_S) \|^2_2.
\end{align*}

Now, we introduce two additional quantities,
\begin{align*}
	& \Delta \mathbf{D}^n_{SS} = \mathbf{D}^n_{SS}(\boldsymbol{\alpha}^*_S+b \mathbf{u}_S) -
	\mathbf{D}^n_{SS}(\boldsymbol{\alpha}^*_S)\\
	& \Delta \mathbf{X}^n_S = \mathbf{X}^n_S(\boldsymbol{\alpha}^*_S+b \mathbf{u}_S) -
	\mathbf{X}^n_S(\boldsymbol{\alpha}^*_S),
\end{align*}
and rewrite $q$ as
\begin{eqnarray*}
& q & = \min_{\mathbf{u}_S} \big[ \mathbf{u}_S^\top
	\mathbf{D}^n_{SS}(\boldsymbol{\alpha}^*_S) \mathbf{u}_S
	+ n^{-1} \| \mathbf{u}_S^\top  \mathbf{X}^n_S(\boldsymbol{\alpha}^*_S) \|^2_2\\
& & + n^{-1}\|\mathbf{u}_S^\top  \Delta \mathbf{X}^n_S \|^2_2 + \mathbf{u}_S^\top  \Delta
\mathbf{D}^n_{SS} \mathbf{u}_S\\
& & + 2 n^{-1} \langle \mathbf{u}_S^\top  \mathbf{X}^n_S(\boldsymbol{\alpha}^*_S),
\mathbf{u}_S^\top  \Delta \mathbf{X}^n_S \rangle \big].
\end{eqnarray*}

Next, we use dependency condition,
\begin{multline*} \label{equ:proof-of-lemma6-inequality-on-q}
	q \geq C_{\min} B^2 - \max_{\mathbf{u}_S} |\underbrace{\mathbf{u}_S^\top
	\Delta \mathbf{D}^n_{SS} \mathbf{u}_S}_{T_1}| \\
- \max_{\mathbf{u}_S} 2 |\underbrace{n^{-1}\langle \mathbf{u}_S^\top
\mathbf{X}^n_S(\boldsymbol{\alpha}^*_S),  \mathbf{u}_S^\top  \Delta \mathbf{X}^n_S
\rangle}_{T_2}|,
\end{multline*}
and proceed to bound $T_1$ and $T_2$ separately. First, we bound $T_1$ using the Lipschitz condition,
\begin{eqnarray*}
& | T_1 | & = | \sum_{k\in S} u_k^2 [ \mathbf{D}^n_k(\boldsymbol{\alpha}^*_S+b
 \mathbf{u}_S)-\mathbf{D}^n_k(\boldsymbol{\alpha}^*_S)] | \\
& & \leq \sum_{k\in S} u_k^2 k_2 \|b \mathbf{u}_S \|_2\\
& & \leq k_2 B^3.
\end{eqnarray*}

Then, we use the dependency condition, the Lipschitz condition and the Cauchy-Schwartz inequality to bound $T_2$,
\begin{eqnarray*}
& T_2 & \leq  \frac{1}{\sqrt{n}} \| \mathbf{u}_S^\top
\mathbf{X}^n_S(\boldsymbol{\alpha}^*_S) \|_2  \frac{1}{\sqrt{n}}\|\mathbf{u}_S^\top  \Delta
\mathbf{X}^n_S\|_2 \\
& & \leq \sqrt{C_{\max}} B  \frac{1}{\sqrt{n}} \|\mathbf{u}_S^\top  \Delta
\mathbf{X}^n_S\|_2 \\
& & \leq \sqrt{C_{\max}} B  \|\mathbf{u}_S\|_2 \frac{1}{\sqrt{n}} |\|\Delta
\mathbf{X}^n_S\||_2\\
& & \leq \sqrt{C_{\max}} B^2 k_1 \|b \mathbf{u}_S\|_2\\
& & \leq k_1 \sqrt{C_{\max}} B^3,
\end{eqnarray*}
where we note that applying the Lipschitz condition implies assuming $B < \frac{\alpha_{min}}{2}$. Next, we
incorporate the bounds of $T_1$ and $T_2$ to lower bound $q$,
\begin{equation}
 q \geq C_{\min} B^2 - (k_2 + 2 k_1 \sqrt{C_{\max}}) B^3.
\end{equation}

Now, we set $B = K \lambda_n \sqrt{d}$, where $K$ is a constant that we will set later in the proof, and select the regularization parameter $\lambda_n$ to satisfy
$\lambda_n \sqrt{d} \leq 0.5 C_{\min} / K(k_2 + 2 k_1 \sqrt{C_{\max})}$. Then,
\begin{eqnarray*}
& G(\mathbf{u}_S) & \geq - 4^{-1}\lambda_n \sqrt{d} B  + 0.5 C_{\min} B^2  -
\lambda_n \sqrt{d} B \\
& & \geq B(0.5 C_{\min} B - 1.25 \lambda_n \sqrt{d}) \\
& & \geq B(0.5 C_{\min} K \lambda_n \sqrt{d} - 1.25 \lambda_n \sqrt{d}).
\end{eqnarray*}

In the last step, we set the constant $K = 3 C_{\min}^{-1}$, and we have
\begin{equation*}
G(\mathbf{u}_S)\geq 0.25 \lambda_n \sqrt{d} > 0,
\end{equation*}
as long as
\begin{align*}
& \sqrt{d} \lambda_n \leq \frac{C_{\min}^2}{6 (k_2 + 2 k_1 \sqrt{C_{\max}})}\\
& \alpha^*_{\min} \geq \frac{6 \lambda_n \sqrt{d}}{C_{\min}}.
\end{align*}

Finally, convexity of $G(\mathbf{u}_S)$ yields
\begin{equation*}
	\| \boldsymbol{\hat{\alpha}}_S - \boldsymbol{\alpha}^*_S \|_2 \leq 3 \lambda_n
	\sqrt{d} / C_{\min}
	\leq
	\frac{\alpha^*_{\min}}{2}.
\end{equation*}
\begin{figure*}[t]
	\centering
	\subfigure[Chain ($d_i=1$)]{\includegraphics[width=0.32\textwidth]{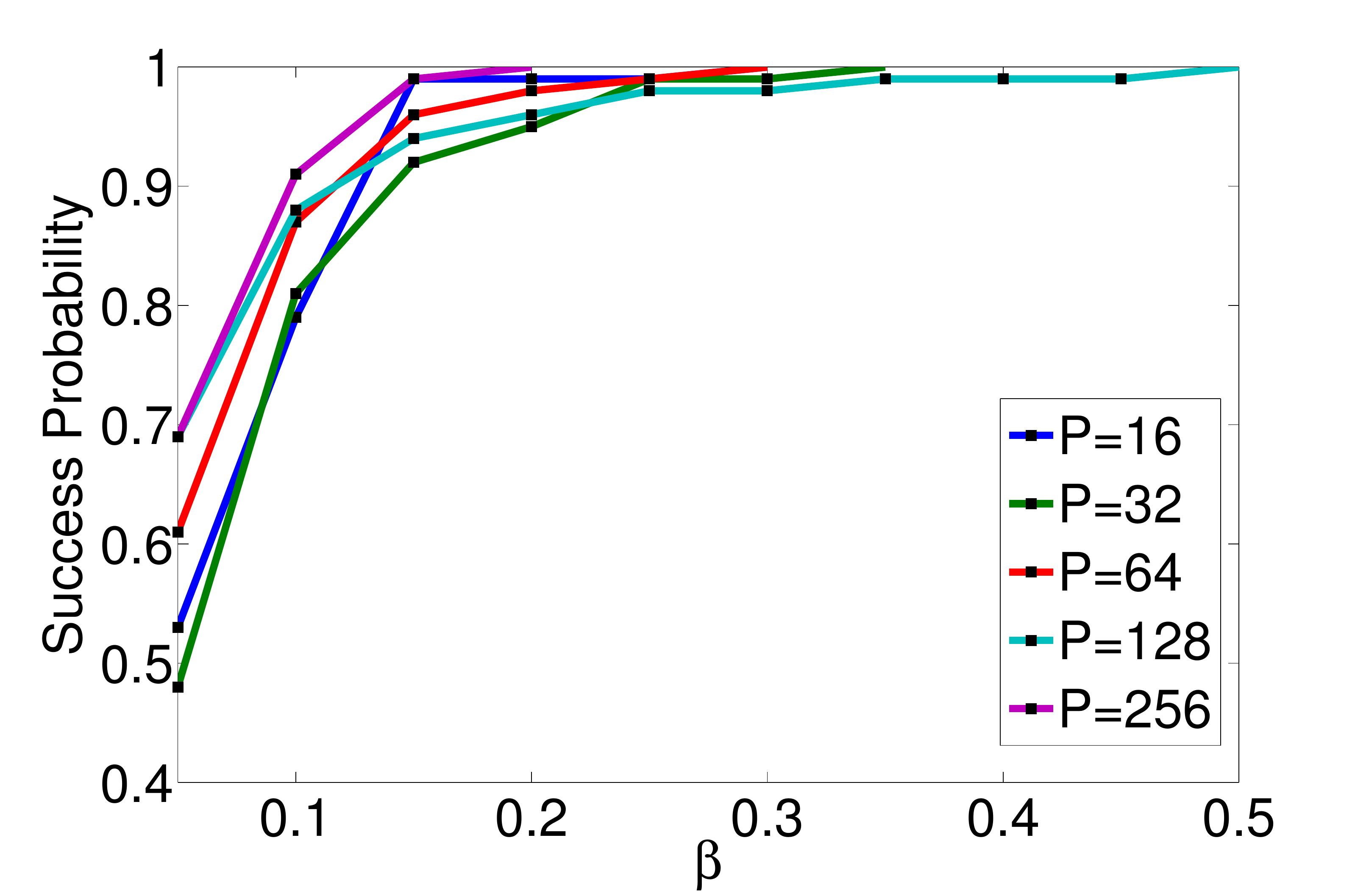}\label{fig:chain-p}}
	\subfigure[Stars with different \# of leaves ($d_i=1$)]{\includegraphics[width=0.32\textwidth]{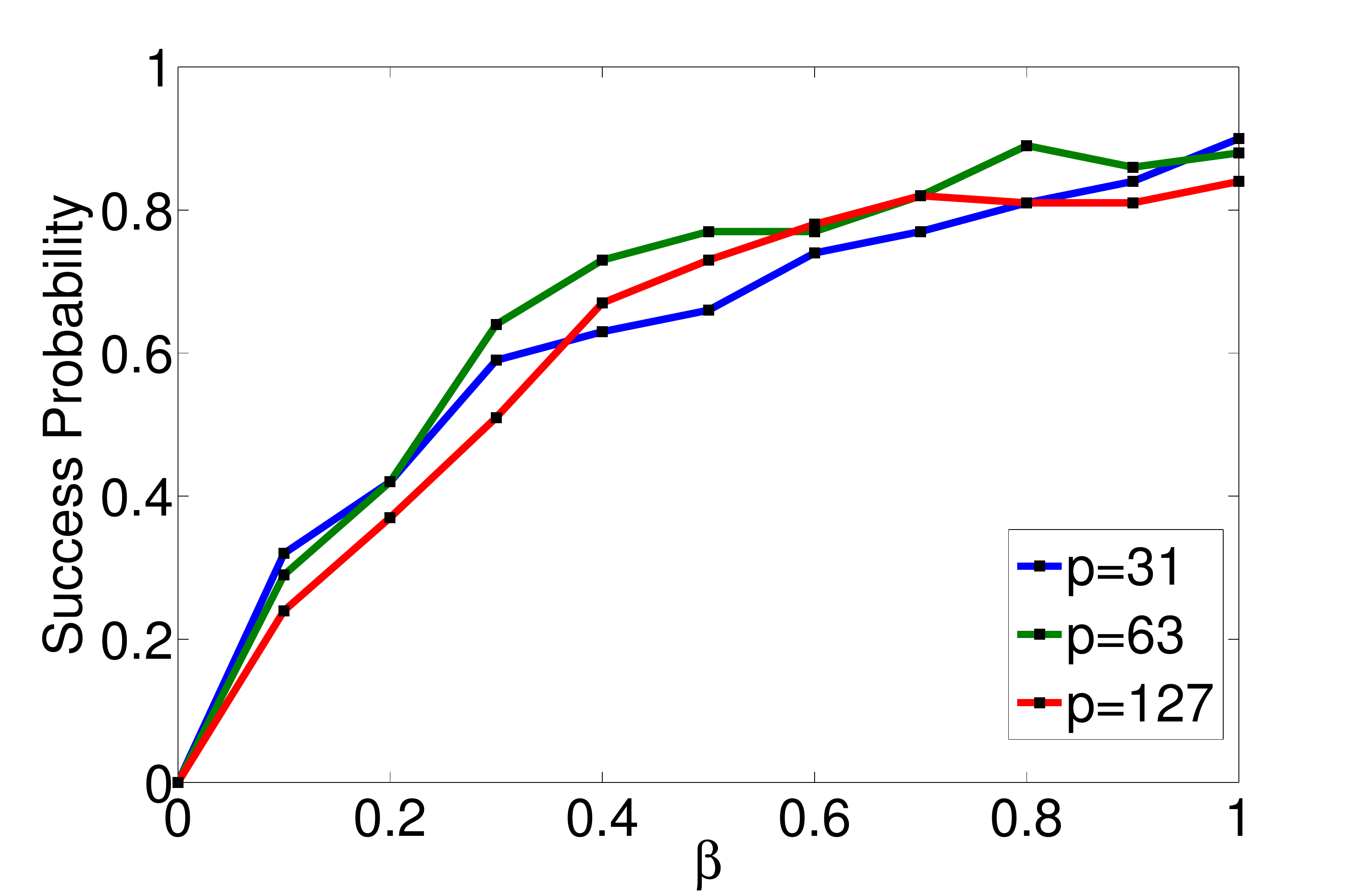}\label{fig:star-p}}
	\subfigure[Tree ($d_i=3$)]{\includegraphics[width=0.32\textwidth]{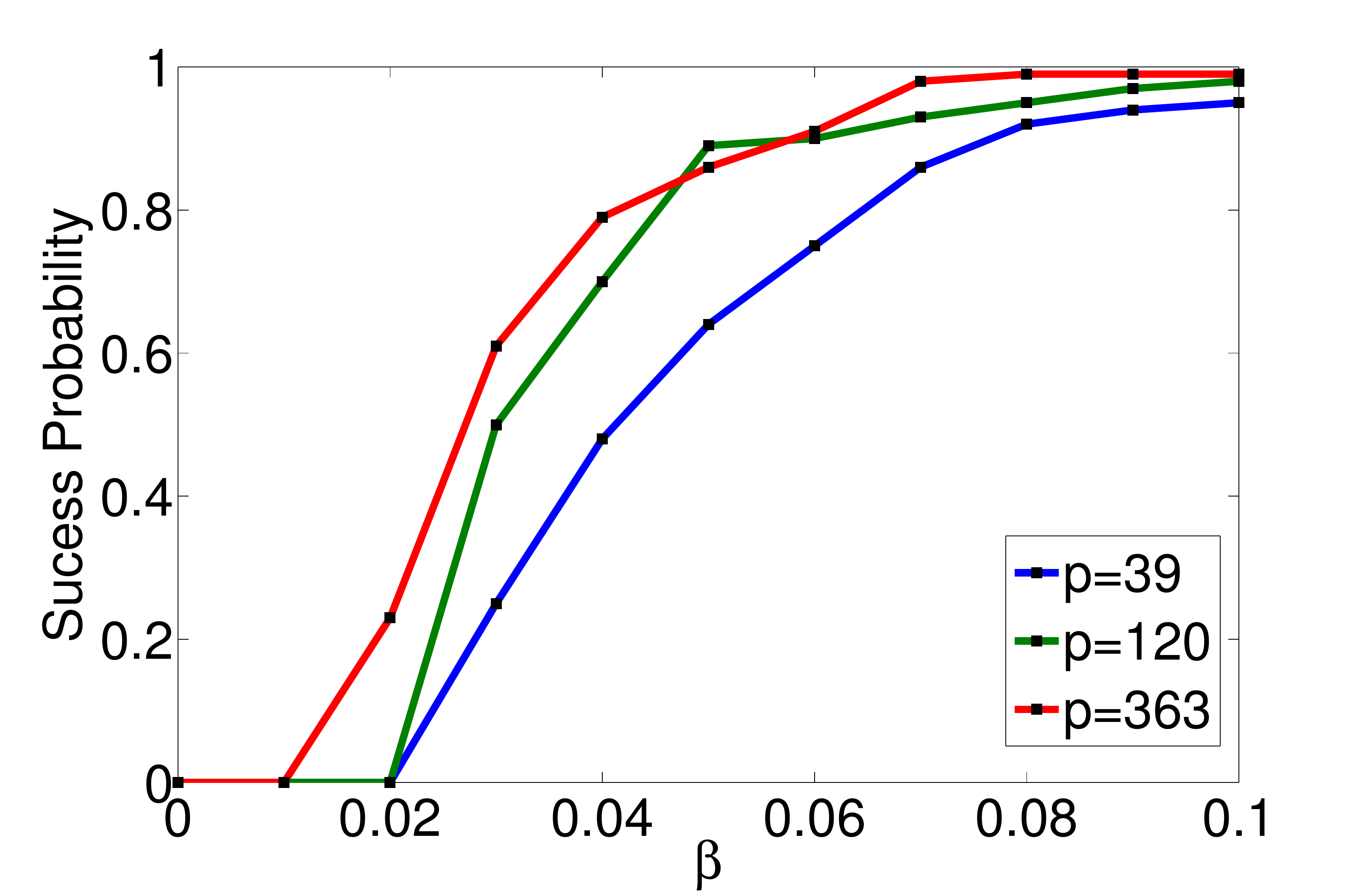}\label{fig:tree-p}}
 	\caption{Success probability vs. \# of cascades. Different super-neighborhood sizes $p_i$.} \label{fig:additional-p}
\end{figure*}

\section{Proof of Lemma~\ref{lemma:gradient}} \label{app:proof-gradient-hoeffding}
Define $z_j^c = [\nabla g(\mathbf{t}^c; \boldsymbol{\alpha}^*)]_j$ and $z_j = \frac{1}{n} \sum_c z_j^c$.
Now, using\- the KKT conditions and condition 4 (Boundedness), we have that $\mu^*_j = \mathbb{E}_c\{z_j^c\}$ and $|z_j^c| \leq k_3$,
respectively.
Thus, Hoeffding'{}s inequality yields
\begin{multline*}
	P(|z_j - \mu^*_j|>\frac{ \lambda_n \varepsilon}{4(2-\varepsilon)})  \\
	\leq 2\exp\left(- \frac{n \lambda^2_n
	\varepsilon^2}{32k^2_3\left(2-\varepsilon\right)^2}\right),
\end{multline*}
and then,
\begin{multline*}
	P(\|\boldsymbol{z} - \boldsymbol{\mu}^*\|_{\infty}>\frac{ \lambda_n
	\varepsilon}{4(2-\varepsilon)})
	\\
	\leq 2\exp\left(- \frac{n \lambda^2_n
	\varepsilon^2}{32k^2_3\left(2-\varepsilon)\right)^2}+ \log p \right).
\end{multline*}

\section{Proof of Lemma~\ref{lemma:taylor-error}} \label{app:proof-taylor-error}
We start by factorizing the Hessian matrix, using Eq.~\ref{eq:consistency-hessian-matrix-factorization},
\begin{equation*}
	R^n_j = \left[\nabla^2 \ell^n(\boldsymbol{\bar{\alpha}}_j) - \nabla^2
	\ell^n(\boldsymbol{\alpha^*}) \right]_j^\top  (\boldsymbol{\hat{\alpha}} - \boldsymbol{\alpha^*}) =  \omega^n_j + \delta^n_j,
\end{equation*}
where,
\begin{align*}
\omega^n_j   &= \big[\mathbf{D}^n(\boldsymbol{\bar{\alpha}}_j) -
	\mathbf{D}^n(\boldsymbol{\alpha^*})\big]_j^\top  (\boldsymbol{\hat{\alpha}} -
	\boldsymbol{\alpha^*}) \\
\delta^n_j &= \frac{1}{n} \boldsymbol{V}^n_j (\boldsymbol{\hat{\alpha}} -
	\boldsymbol{\alpha^*}) \\
\boldsymbol{V}^n_j &=  [\mathbf{X}^n(\boldsymbol{\bar{\alpha}}_j)]_j
	\mathbf{X}^n(\boldsymbol{\bar{\alpha}}_j)^\top - [\mathbf{X}^n(\boldsymbol{\alpha^*})]_j
	\mathbf{X}^n(\boldsymbol{\alpha^*})^\top.
\end{align*}

Next, we proceed to bound each term separately.
Since $[\boldsymbol{\bar{\alpha}}_j]_S = \theta_j \boldsymbol{\hat{\alpha}}_S + (1-\theta_j)\boldsymbol{\alpha}^{*}_S$ where $\theta_j \in [0,1]$, and
$\| \boldsymbol{\hat{\alpha}}_S - \boldsymbol{\alpha}^{*}_S  \|_{\infty} \leq \frac{\alpha^*_{\min}}{2}$ (Lemma~\ref{lemma:alphas-s}), it holds that
$[\boldsymbol{\bar{\alpha}}_j]_S \geq \frac{\alpha^*_{\min}}{2}$. Then, we can use condition 3 (Lipschitz Continuity) to bound $\omega^n_j$.
\begin{eqnarray} \label{equ:appendix-lemma8-proof-omega-bound}
	& |\omega^n_j| & \leq k_1 \| \boldsymbol{\bar{\alpha}}_j -
	\boldsymbol{\alpha^*} \|_2 \|\boldsymbol{\hat{\alpha}} -
	\boldsymbol{\alpha^*} \|_2 \nonumber \\
	& & \leq k_1 \theta_j \|\boldsymbol{\hat{\alpha}} -
	\boldsymbol{\alpha^*}\|_2^2 \nonumber \\
	& & \leq k_1  \|\boldsymbol{\hat{\alpha}} -
	\boldsymbol{\alpha^*}\|_2^2.
\end{eqnarray}

However, bounding term $\delta^n_j$ is more difficult. Let us start by rewriting $\delta^n_j$ as follows.
\begin{equation*}
	\delta^n_j = \left( \boldsymbol{\Lambda}_1 + \boldsymbol{\Lambda}_2 + \boldsymbol{\Lambda}_3 \right) (\boldsymbol{\hat{\alpha}} - \boldsymbol{\alpha^*}),
\end{equation*}
where,
\begin{align*}
   & \boldsymbol{\Lambda}_1 =  [\mathbf{X}^n(\boldsymbol{\alpha^*})]_j
	(\mathbf{X}^n(\boldsymbol{\bar{\alpha}}_j)^\top  -
	\mathbf{X}^n(\boldsymbol{\alpha^*})^\top ) \\
   & \boldsymbol{\Lambda}_2 = \{ [\mathbf{X}^n(\boldsymbol{\bar{\alpha}}_j)]_j
	- [\mathbf{X}^n(\boldsymbol{\alpha^*})]_j
	\}(\mathbf{X}^n(\boldsymbol{\bar{\alpha}}_j)^\top  -
	\mathbf{X}^n(\boldsymbol{\alpha^*})^\top )\\
   & \boldsymbol{\Lambda}_3 = \big(  [\mathbf{X}^n(\boldsymbol{\bar{\alpha}}_j)]_j -
    [\mathbf{X}^n(\boldsymbol{\alpha^*})]_j
	\big)\mathbf{X}^n(\boldsymbol{\alpha^*})^\top.
\end{align*}

Next, we bound each term separately. For the first term, we first apply Cauchy inequality,
\begin{multline*}
	 | \boldsymbol{\Lambda}_1(\boldsymbol{\hat{\alpha}} -
	\boldsymbol{\alpha^*})| 	\leq \|[\mathbf{X}^n(\boldsymbol{\alpha^*})]_j\|_2
	\\
	\times|\| \mathbf{X}^n(\boldsymbol{\bar{\alpha}}_j)^\top  -
	\mathbf{X}^n(\boldsymbol{\alpha^*})^\top  \||_2 \|\boldsymbol{\hat{\alpha}} -
	\boldsymbol{\alpha^*}\|_2,
\end{multline*}
and then use condition 3 (Lipschtiz Continuity) and 4 (Boundedness),
\begin{eqnarray*}
	& |\boldsymbol{\Lambda}_1(\boldsymbol{\hat{\alpha}} -
	\boldsymbol{\alpha^*}) |
	& \leq  n k_4 k_1 \|\boldsymbol{\bar{\alpha}}_j -
	\boldsymbol{\alpha^*} \|_2 \|\boldsymbol{\hat{\alpha}} -
	\boldsymbol{\alpha^*}\|_2 \\
	& & \leq n k_4 k_1\|\boldsymbol{\hat{\alpha}} -
	\boldsymbol{\alpha^*}\|_2^2.
\end{eqnarray*}

For the second term, we also start by applying Cauchy inequality,
\begin{multline*}
	| \boldsymbol{\Lambda}_2(\boldsymbol{\hat{\alpha}} -
	\boldsymbol{\alpha^*})| 	\leq \| [\mathbf{X}^n(\boldsymbol{\bar{\alpha}}_j)]_j
	- [\mathbf{X}^n(\boldsymbol{\alpha^*})]_j\|_2
	\\
	\times|\| \mathbf{X}^n(\boldsymbol{\bar{\alpha}}_j)^\top  -
	\mathbf{X}^n(\boldsymbol{\alpha^*})^\top  \||_2 \|\boldsymbol{\hat{\alpha}} -
	\boldsymbol{\alpha^*}\|_2,
\end{multline*}
and then use condition 3 (Lipschtiz Continuity),
\begin{equation*}
| \boldsymbol{\Lambda}_2(\boldsymbol{\hat{\alpha}} -
	\boldsymbol{\alpha^*})| \leq n k_1^2 \|\boldsymbol{\hat{\alpha}} -
	\boldsymbol{\alpha^*}\|_2^2.
\end{equation*}

Last, for third term, once more we start by applying Cauchy inequality,
\begin{multline*}
	| \boldsymbol{\Lambda}_3(\boldsymbol{\hat{\alpha}} -
	\boldsymbol{\alpha^*})| 	\leq \| [\mathbf{X}^n(\boldsymbol{\bar{\alpha}}_j)]_j
	- [\mathbf{X}^n(\boldsymbol{\alpha^*})]_j\|_2
	\\
	\times|\|
	\mathbf{X}^n(\boldsymbol{\alpha^*})^\top  \||_2 \|\boldsymbol{\hat{\alpha}} -
	\boldsymbol{\alpha^*}\|_2,
\end{multline*}
and then apply condition 1 (Dependency Condition) and condition 3 (Lipschitz Continuity),
\begin{equation*}
| \boldsymbol{\Lambda}_3(\boldsymbol{\hat{\alpha}} -
	\boldsymbol{\alpha^*})| \leq n k_1 \sqrt{C_{\max}} \|\boldsymbol{\hat{\alpha}}
	- \boldsymbol{\alpha^*}\|_2^2
\end{equation*}

Now, we combine the bounds,
\begin{equation*}
	\| \boldsymbol{R}^n \|_{\infty} \leq K \|\boldsymbol{\hat{\alpha}}
	- \boldsymbol{\alpha^*}\|_2^2,
\end{equation*}
where
\begin{equation*}
	K = k_1 + k_4 k_1 + k_1^2 + k_1 \sqrt{C_{\max}}.
\end{equation*}

Finally, using Lemma~\ref{lemma:alphas-s} and selecting the regularization parameter $\lambda_n$ to satisfy $\lambda_n d \leq C_{\min}^{2} \frac{\varepsilon}{36 K (2-\varepsilon)}$
yields:
\begin{eqnarray*}
	& \| \boldsymbol{R}^n \|_{\infty}/\lambda_n  & \leq 3 K \lambda_n d  /
	C_{\min}^2 \\
	& & \leq \frac{\varepsilon}{4(2-\varepsilon)}
\end{eqnarray*}

\section{Proof of Lemma~\ref{lemma:dependency-n}} \label{app:dependency-n}
We will first bound the difference in terms of nuclear norm between the population Fisher information matrix $\mathcal{Q}_{SS}$ and the sample mean cascade
log-likelihood $\mathcal{Q}^n_{SS}$.
Define $z_{jk}^c = [\nabla^2 g(\mathbf{t}^c; \boldsymbol{\alpha}^*) - \nabla^2 \ell^n(\boldsymbol{\alpha}^*)]_{jk}$ and $z_{jk} = \frac{1}{n}\sum_{c=1}^n z_{jk}^{c}$.
Then, we can express the difference between the population Fisher information matrix $\mathcal{Q}_{SS}$ and the sample mean cascade log-likelihood
$\mathcal{Q}^n_{SS}$ as:
\begin{multline*}
 |\| \mathcal{Q}^n_{SS}(\boldsymbol{\alpha^*}) -
\mathcal{Q}^{*}_{SS}(\boldsymbol{\alpha^*}) \||_2 \\
\leq |\| \mathcal{Q}^n_{SS}(\boldsymbol{\alpha^*}) -
\mathcal{Q}^{*}_{SS}(\boldsymbol{\alpha^*}) \||_F \\
= \sqrt{\sum_{j=1}^d \sum_{k=1}^d (z_{ik})^2}.
\end{multline*}

Since $|z_{jk}^{(c)}| \leq 2 k_5$ by condition 4, we can apply Hoeff\-ding'{}s inequality to each $z_{jk}$,
\begin{equation} \label{eq:appendix-lemma9-bound-on-z}
	P(|z_{jk}| \geq \beta ) \leq 2 \exp\left(-\frac{\beta^2 n }{ 8 k_5^2} \right),
\end{equation}
and further,
\begin{multline}
\label{eq:appendix-lemma9-bound-on-spectral-norm}
P(|\|\mathcal{Q}^n_{SS}(\boldsymbol{\alpha^*}) - \mathcal{Q}^{*}_{SS}(\boldsymbol{\alpha^*}) \||_2 \geq \delta)\\
\leq 2 \exp\big( -K \frac{\delta^2 n}{d^2} + 2 \log d   \big)
\end{multline}
where $\beta^2 = \delta^2 / d^2$.
Now, we bound the maximum eigenvalue of $\mathcal{Q}^n_{SS}$ as fo\-llows:
\begin{eqnarray*}
& \Lambda_{\max}(\mathcal{Q}^n_{SS}) & = \max_{\|x\|_2 = 1} x^\top
\mathcal{Q}^n_{SS} x\\
& & = \max_{\|x\|_2 = 1} \{ x^\top  \mathcal{Q}^*_{SS} x + x^\top  ( \mathcal{Q}^n_{SS}
- \mathcal{Q}^*_{SS}) x\}\\
& & \leq y^\top  \mathcal{Q}^*_{SS} y + y^\top  ( \mathcal{Q}^n_{SS}
- \mathcal{Q}^*_{SS}) y,
\end{eqnarray*}
where $y$ is unit-norm maximal eigenvector of $\mathcal{Q}^*_{SS}$. Therefore,
\begin{equation*}
	\Lambda_{\max}(\mathcal{Q}^n_{SS})\leq \Lambda_{\max}(\mathcal{Q}^*_{SS}) + |\| \mathcal{Q}^n_{SS} -
\mathcal{Q}^{*}_{SS} \||_2,
\end{equation*}
and thus,
\begin{multline*}
	P\big(\Lambda_{\max}(\mathcal{Q}^n_{SS})\geq C_{\max} + \delta \big) \\
	\leq \exp\left( -K \frac{\delta^2 n}{d^2} + 2 \log d  \right).
\end{multline*}

Reasoning in a similar way, we bound the minimum eigenvalue of  $\mathcal{Q}^n_{SS}$:
\begin{multline*}
	P\big(\Lambda_{\min}(\mathcal{Q}^n_{SS})\leq C_{\min} - \delta \big) \\
	\leq \exp\left( -K \frac{\delta^2 n}{d^2} + 2 \log d  \right)
\end{multline*}
\begin{figure}[t]
	\centering
	\includegraphics[width=0.3\textwidth]{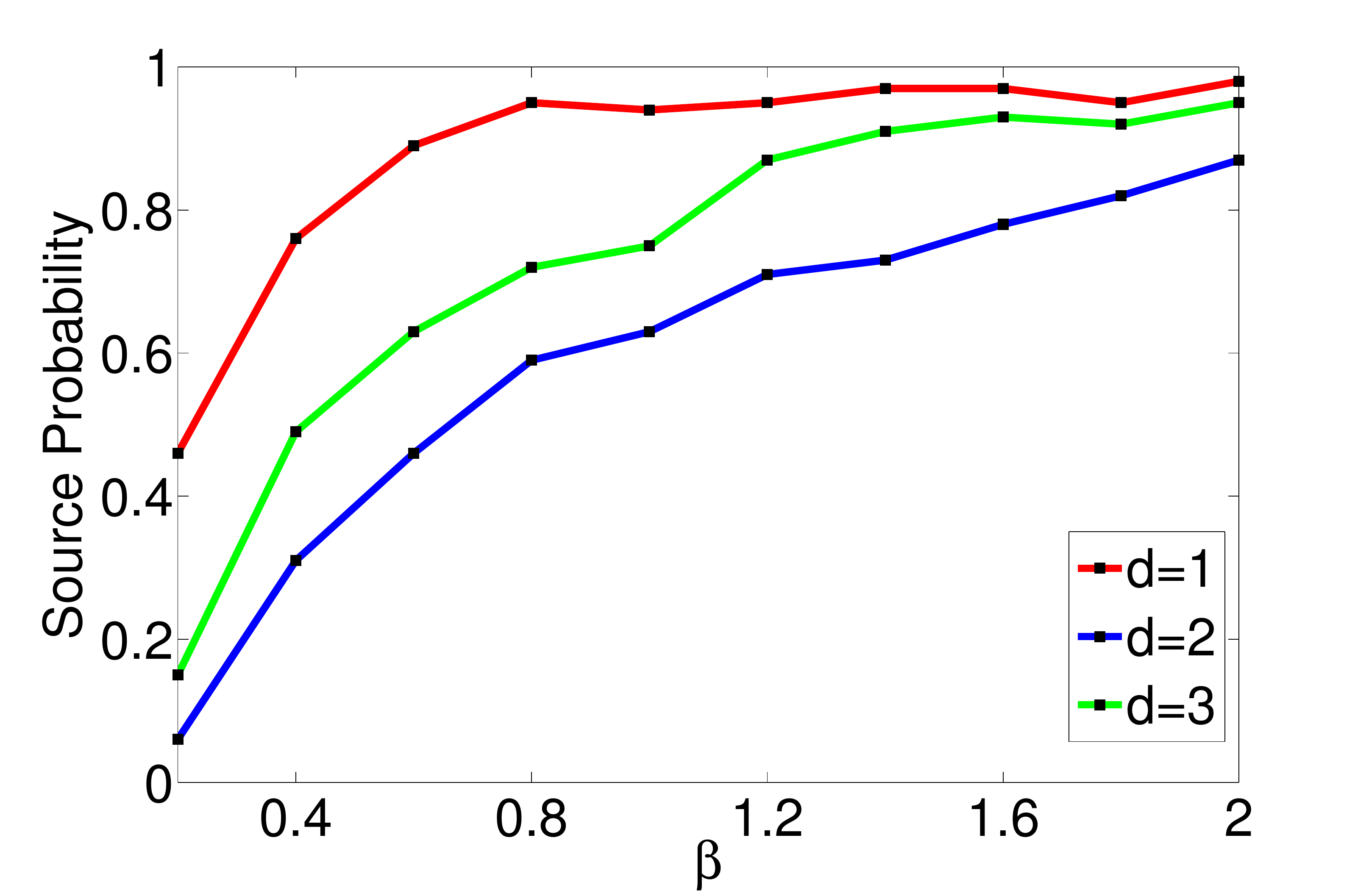}
 	\caption{Success probability vs. \# of cascades. Different in-degrees $d_i$.} \label{fig:additional-d}
\end{figure}

\section{Proof of Lemma~\ref{lemma:coherence-n}}  \label{app:coherence-n}
We start by decomposing $\mathcal{Q}^n_{S^cS}(\boldsymbol{\alpha^*}) (\mathcal{Q}^n_{S^cS}(\boldsymbol{\alpha^*}))^{-1}$ as fo\-llows:
\begin{equation*}
\mathcal{Q}^n_{S^cS}(\boldsymbol{\alpha^*}) (\mathcal{Q}^n_{S^cS}(\boldsymbol{\alpha^*}))^{-1} = A_1 + A_2 + A_3 + A_4,
\end{equation*}
where,
\begin{align*}
	& A_1 = \mathcal{Q}^*_{S^cS} [ (\mathcal{Q}^n_{S^cS})^{-1} -
	(\mathcal{Q}^*_{S^cS})^{-1} ], \\
	& A_2 = [\mathcal{Q}^n_{S^cS} - \mathcal{Q}^*_{S^cS}]
	[ (\mathcal{Q}^n_{S^cS})^{-1} -
	(\mathcal{Q}^*_{S^cS})^{-1} ]\\
	 & A_3 = [\mathcal{Q}^n_{S^cS} -
	\mathcal{Q}^*_{S^cS}](\mathcal{Q}^*_{SS})^{-1},\\
	& A_4 = \mathcal{Q}^*_{S^cS}(\mathcal{Q}^*_{SS})^{-1},
\end{align*}
$\mathcal{Q}^* = \mathcal{Q}^*(\boldsymbol{\alpha^*})$ and $\mathcal{Q}^n = \mathcal{Q}^n(\boldsymbol{\alpha^*})$.
Now, we bound each term separately. The fourth term, $A_4$, is the easiest to bound, using simply the incoherence condition:
\begin{equation*}
	|\| A_4 \||_{\infty} \leq 1- \varepsilon.
\end{equation*}
\begin{figure*}[t]
	\centering
	\subfigure[Kronecker hierarchical, \expo]{\makebox[4.2cm][c]{\includegraphics[width=0.23\textwidth]{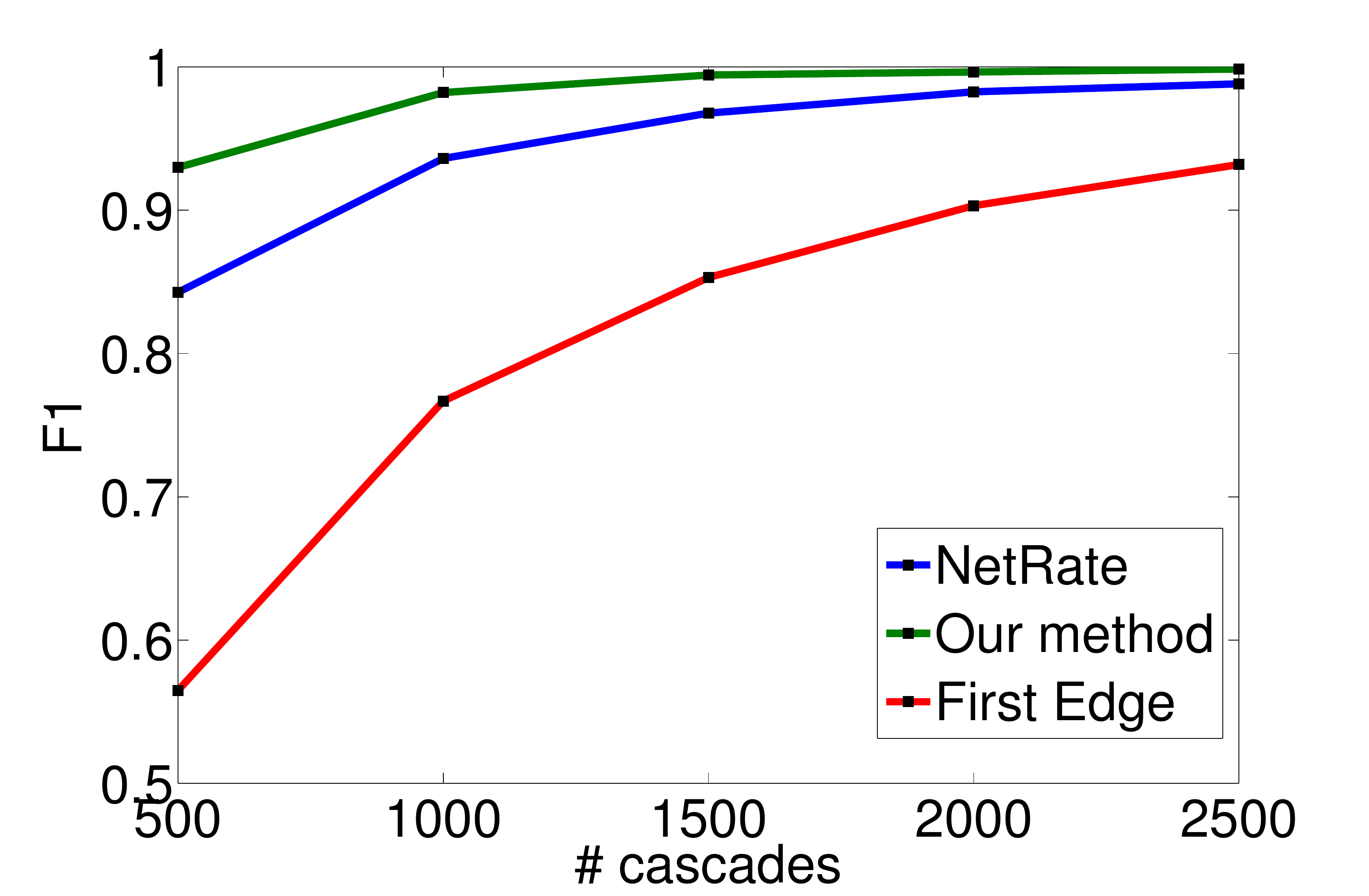}\label{fig:hierarchical-exp}}}
	\subfigure[Kronecker hierarchical, \ray]{\makebox[4.2cm][c]{\includegraphics[width=0.23\textwidth]{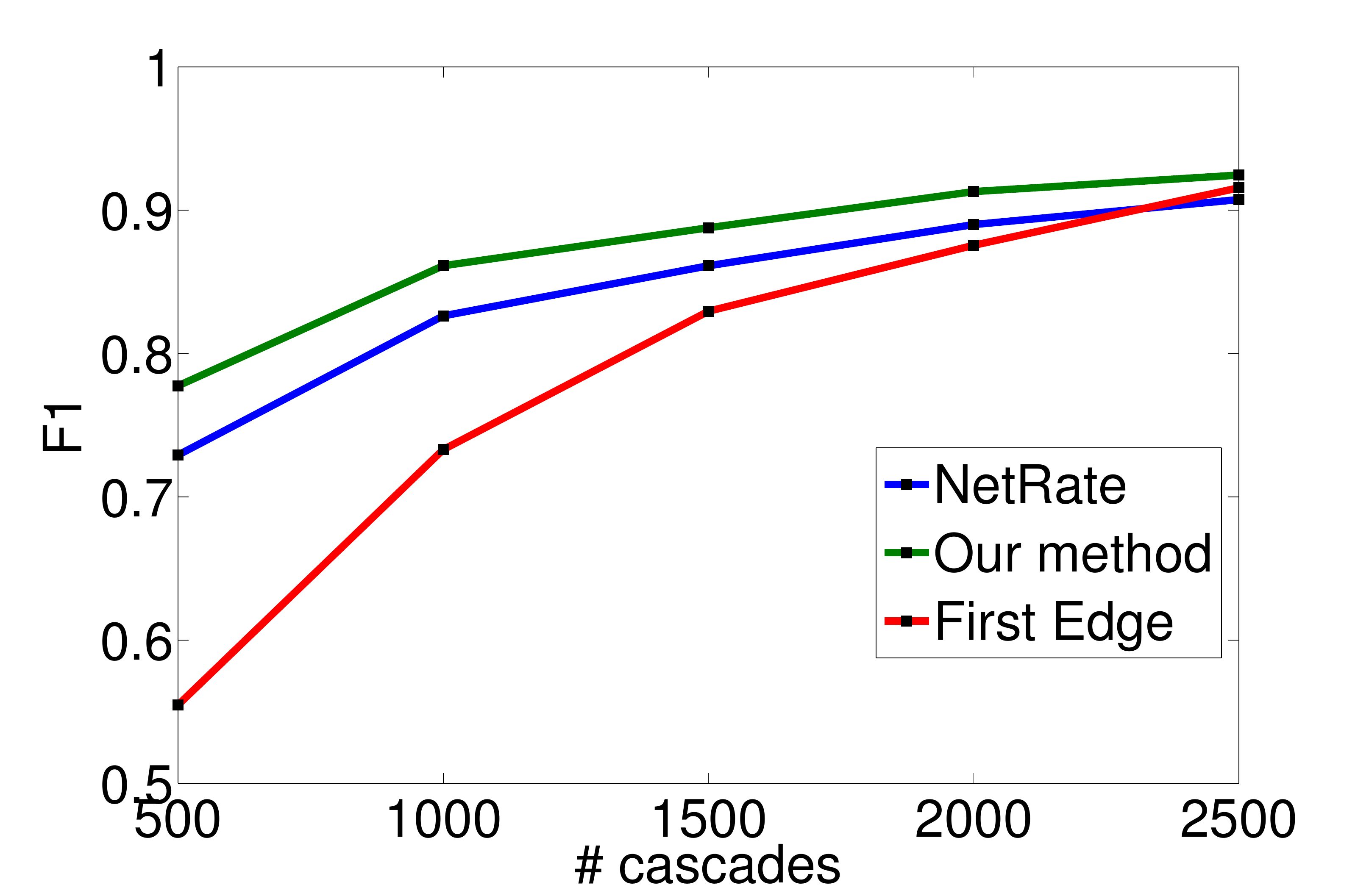}\label{fig:hierarchical-ray}}} 
	\subfigure[Forest Fire, \pow]{\includegraphics[width=0.23\textwidth]{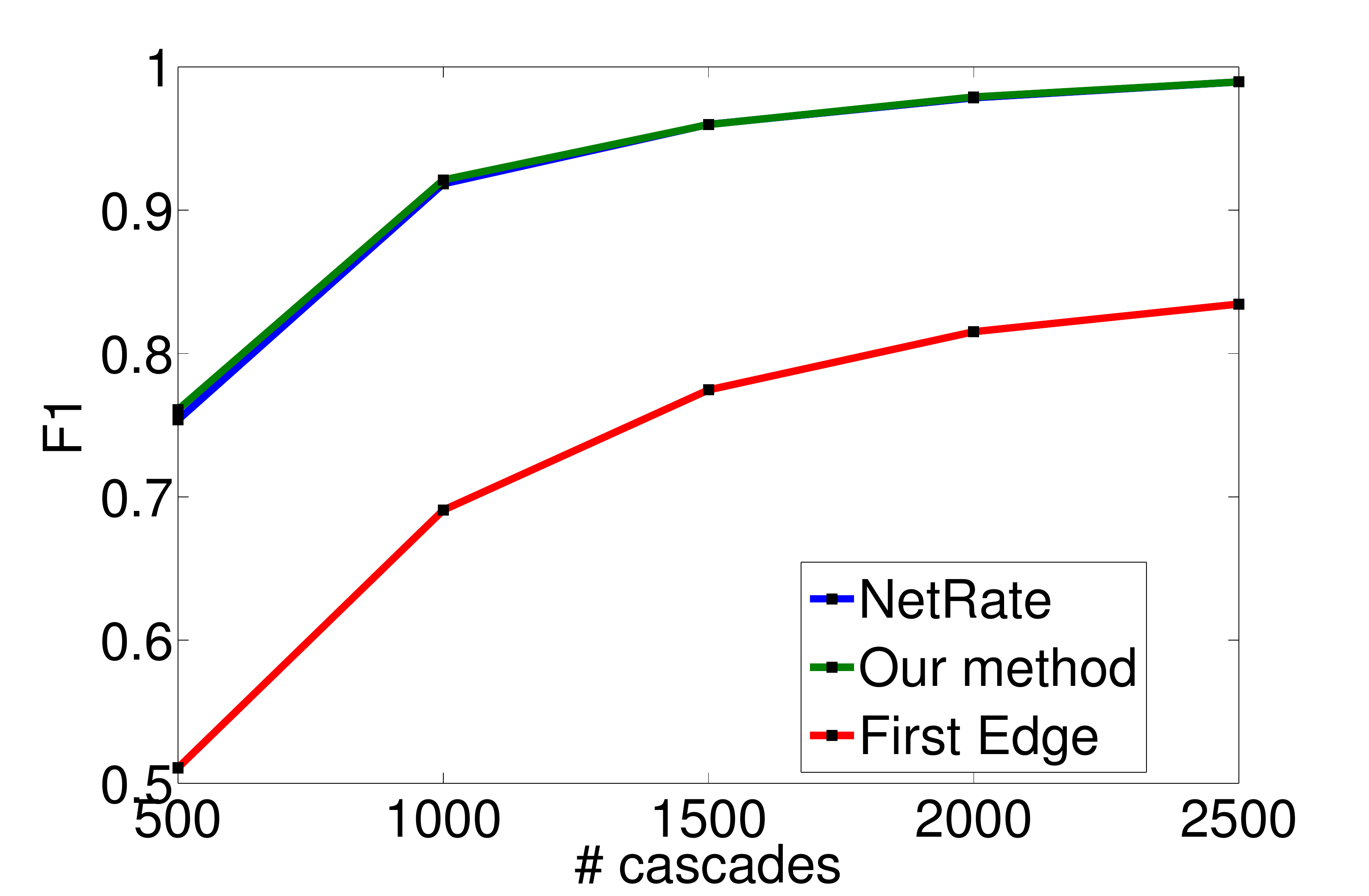}\label{fig:ff-pow}}
	\subfigure[Forest Fire, \ray]{\includegraphics[width=0.23\textwidth]{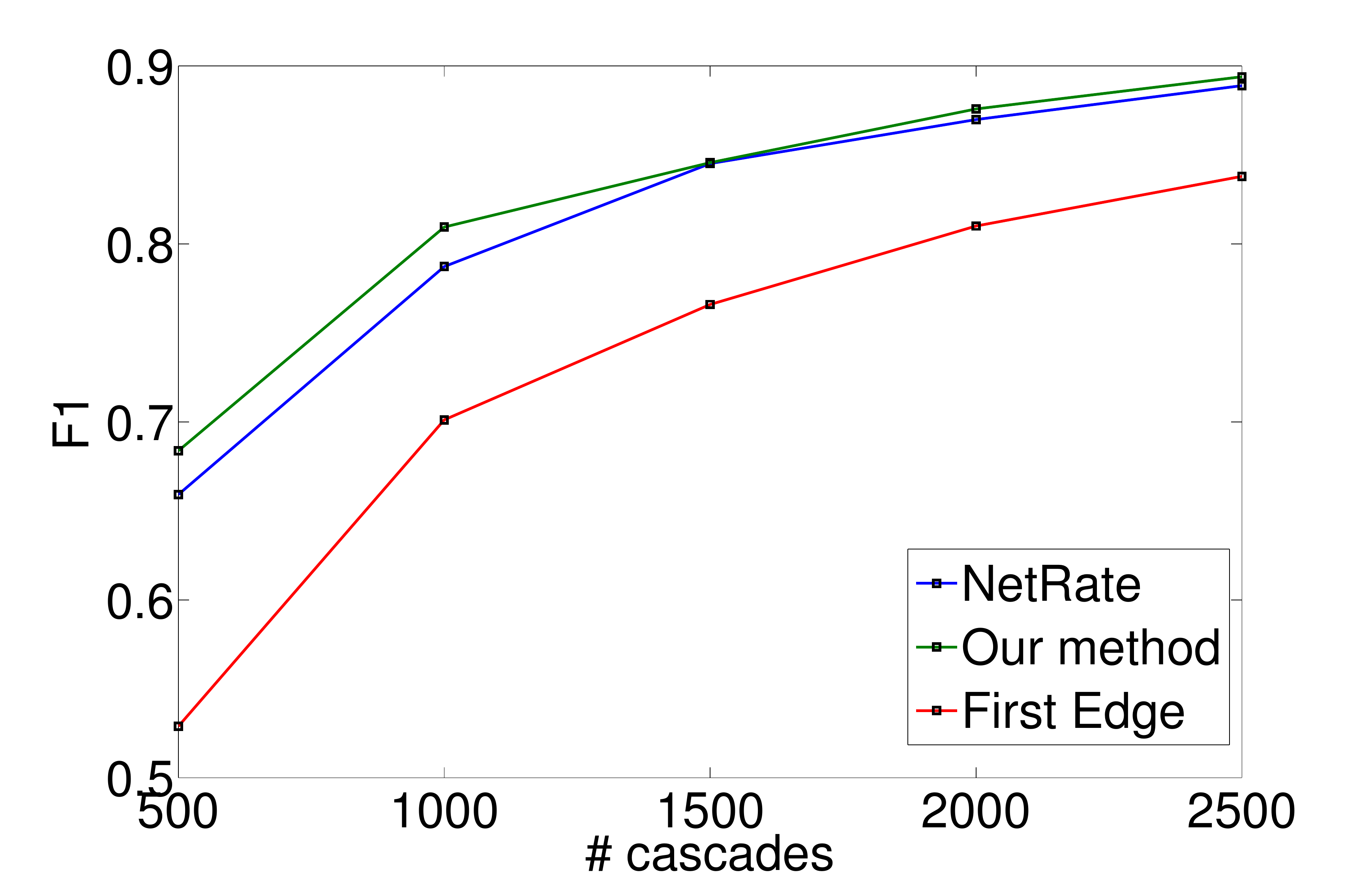}\label{fig:ff-rayleigh}}
 	\caption{$F_{1}$-score vs. \# of cascades.} \label{fig:additional-comparison}
\end{figure*}

To bound the other terms, we need the following lemma:
\begin{lemma}
	For any $\delta \geq 0$ and constants $K$ and $K^{\prime}$, the following bounds hold:
	\begin{multline} \label{eq:appendix-lemma12-first-inequality}
		 P[|\|\mathcal{Q}^n_{S^cS} - \mathcal{Q}^*_{S^cS}\||_{\infty} \geq \delta]\\
		\leq 2\exp\left(-K\frac{n\delta^2}{d^2}+\log d + log(p-d)\right)
	\end{multline}
	\begin{multline} \label{eq:appendix-lemma12-second-inequality}
		P[|\|\mathcal{Q}^n_{SS} - \mathcal{Q}^*_{SS}\||_{\infty} \geq \delta] \\
		\leq 2\exp\left(-K\frac{n\delta^2}{d^2}+2 \log d \right)
	\end{multline}
	\begin{multline} \label{eq:appendix-lemma12-third-inequality}
		P[|\|(\mathcal{Q}^n_{SS})^{-1} - (\mathcal{Q}^*_{SS})^{-1}\||_{\infty} \geq
		\delta]
		\\
		\leq 4 \exp\left(-K\frac{n\delta}{d^3}- K^{\prime} \log d \right)
	\end{multline}
\end{lemma}

\begin{proof}
	We start by proving the first confidence interval. By definition of infinity norm of a matrix, we have:
	\begin{multline*}
	P[|\|\mathcal{Q}^n_{S^cS} - \mathcal{Q}^*_{S^cS}\||_{\infty} \geq \delta] \\
		= P\big[\max_{j\in S^c} \sum_{k\in S} |z_{jk}| \geq \delta \big]\\
		\leq (p-d) P\big[ \sum_{k\in S} |z_{jk}| \geq \delta \big],
	\end{multline*}
	where $z_{jk} = \left[\mathcal{Q}^n-\mathcal{Q}^*\right]_{jk}$ and, for the last inequality, we used the union bound and the fact that $|S^c| \leq p-d$.
	Furthermore,
	\begin{eqnarray*}
		&  P\big[ \sum_{k\in S} |z_{jk}| \geq \delta \big] & \leq P[\exists k \in S|
		|z_{jk}| \geq \delta/d]  \\
		& & \leq d P[|z_{jk}| \geq \delta/d].
	\end{eqnarray*}
	Thus,
	\begin{equation*}
		P[|\|\mathcal{Q}^n_{S^cS} - \mathcal{Q}^*_{S^cS}\||_{\infty} \geq \delta] \leq
		(p-d) d  P[|z_{jk}| \geq \delta/d].
	\end{equation*}

	At this point, we can obtain the first confidence bound by using Eq.~\ref{eq:appendix-lemma9-bound-on-z} with $\beta = \delta /d$ in the above equation. The proof of
	the second confidence bound is very similar and we omit it for brevity.
	To prove the last confidence bound, we proceed as follows:
	\begin{align*}
		& |\|(\mathcal{Q}^n_{SS})^{-1} - (\mathcal{Q}^*_{SS})^{-1}\||_{\infty} \\
		& = |\|(\mathcal{Q}^n_{SS})^{-1} [\mathcal{Q}^n_{SS} -
		\mathcal{Q}^*_{SS}](\mathcal{Q}^*_{SS})^{-1}\||_{\infty} \\
		& \leq \sqrt{d} |\|(\mathcal{Q}^n_{SS})^{-1} [\mathcal{Q}^n_{SS} -
		\mathcal{Q}^*_{SS}](\mathcal{Q}^*_{SS})^{-1}\||_{2}\\
		& \leq \sqrt{d} |\|(\mathcal{Q}^n_{SS})^{-1}\||_2 |\| \mathcal{Q}^n_{SS} -
		\mathcal{Q}^*_{SS}\||_{2} |\|(\mathcal{Q}^*_{SS})^{-1}\||_{2} \\
		& \leq \frac{\sqrt{d}}{C_{\min}} |\| \mathcal{Q}^n_{SS} -
		\mathcal{Q}^*_{SS}\||_{2} |\|(\mathcal{Q}^n_{SS})^{-1}\||_{2}.
	\end{align*}
	Next, we bound each term of the final expression in the above equation separately.
	The first term can be bounded using Eq.~\ref{eq:appendix-lemma9-bound-on-spectral-norm}:
	\begin{multline*}
		P\big[ |\| \mathcal{Q}^n_{SS} -
		\mathcal{Q}^*_{SS}\||_{2} \geq C_{min}^2 \delta/2 \sqrt{d} \big] \\
		\leq 2 \exp\big(-K\frac{n\delta^2}{d^3}+2 \log d \big),
	\end{multline*}
	The second term can be bounded using Lemma~\ref{lemma:taylor-error}:
	\begin{multline*}
		P\big[|\| (\mathcal{Q}^n_{SS})^{-1}\||_{2} \geq \frac{2}{C_{\min}}\big]\\
		= P\big[ \Lambda_{\min}(\mathcal{Q}^n_{SS}) \leq \frac{C_{\min}}{2}\big] \\
		\leq \exp \left(-K \frac{n}{d^2} + B \log d \right).
	\end{multline*}

	Then, the third confidence bound follows.
\end{proof}

\textit{Control of $A_1$.}
We start by rewriting the term $A_1$ as
\begin{equation*}
A_1 = \mathcal{Q}^*_{S^cS} (\mathcal{Q}^*_{SS})^{-1}[
(\mathcal{Q}^*_{SS})- (\mathcal{Q}^n_{SS}) ](\mathcal{Q}^n_{SS})^{-1},
\end{equation*}
and further,
\begin{multline*}
	 |\| A_1 \||_{\infty}  \leq |\| \mathcal{Q}^*_{S^cS}
	 (\mathcal{Q}^*_{SS})^{-1} \||_{\infty} \\
	 \times |\|(\mathcal{Q}^*_{SS})- (\mathcal{Q}^n_{SS}) \||_{\infty} |\| (\mathcal{Q}^n_{SS})^{-1}
	 \||_{\infty}.
\end{multline*}

Next, using the incoherence condition easily yields:
\begin{multline*}
	 |\| A_1 \||_{\infty}  \leq  (1 - \varepsilon)  |\|(\mathcal{Q}^*_{SS})-
	 (\mathcal{Q}^n_{SS}) \||_{\infty}\\
	 \times \sqrt{d} |\| (\mathcal{Q}^n_{SS})^{-1} \||_{2}
\end{multline*}

Now, we apply Lemma~\ref{lemma:taylor-error} with $\delta = C_{\min}/2$ to have that $|\| (\mathcal{Q}^n_{SS})^{-1} \||_{2} \leq \frac{2}{C_{\min}}$ with probability greater
than $1- \exp(-Kn/d^2 + K^{\prime} \log d)$, and then use Eq.~\ref{eq:appendix-lemma12-third-inequality} with $\delta = \frac{\varepsilon C_{\min}}{12 \sqrt{d}}$ to conclude
that
\begin{equation*}
	P\big[ |\| A_1 \||_{\infty} \geq \frac{\varepsilon}{6} \big] \leq 2 \exp\left(-K
	\frac{n}{d^3} + K^{\prime} \log d\right).
\end{equation*}

\textit{Control of $A_2$.}
We rewrite the term $A_2$ as
\begin{equation*}
	|\|A_2 \||_{\infty} \leq  |\|\mathcal{Q}^n_{S^cS} - \mathcal{Q}^*_{S^cS}\||_{\infty}
	|\| (\mathcal{Q}^n_{SS})^{-1} -
	(\mathcal{Q}^*_{SS})^{-1} \||_{\infty},
\end{equation*}
and then use Eqs.~\ref{eq:appendix-lemma12-first-inequality} and~\ref{eq:appendix-lemma12-second-inequality} with $\delta = \sqrt{\varepsilon/6}$ to conclude that
\begin{multline*}
	P \big[ |\|A_2 \||_{\infty} \geq \frac{\varepsilon}{6} \big] \leq \\
	4 \exp\left(-K \frac{n}{d^3} + \log(p-d) + K^{\prime}\log p \right).
\end{multline*}

\textit{Control of $A_3$.}
We rewrite the term $A_3$ as
\begin{eqnarray*}
	& |\|A_3\||_{\infty} & = \sqrt{d} |\|(\mathcal{Q}^*_{SS})^{-1}\||_{2}
	|\|\mathcal{Q}^n_{S^cS} - \mathcal{Q}^*_{S^cS}\||_{\infty}\\
	& & \leq \frac{\sqrt{d}}{C_{\min}} |\|\mathcal{Q}^n_{S^cS} -
	\mathcal{Q}^*_{S^cS}\||_{\infty}.
\end{eqnarray*}

We then apply Eq.~\ref{eq:appendix-lemma12-first-inequality} with $\delta = \frac{\varepsilon C_{\min}}{6 \sqrt{d}}$ to conclude that
\begin{equation*}
	P\big[ |\|A_3\||_{\infty}\geq \frac{\varepsilon}{6} \big] \leq \exp\left(-K
	\frac{n}{d^3} + \log(p-d) \right),
\end{equation*}
and thus,
\begin{multline*}
	P\big[ |\|\mathcal{Q}^n_{S^cS}
(\mathcal{Q}^n_{SS})^{-1} \||_{\infty} \geq 1- \frac{\varepsilon}{2}  \big]\\
= \mathcal{O}\left(\exp(-K \frac{n}{d^3}+ \log p)\right).
\end{multline*}

\section{Additional experiments} \label{app:experiments}
\xhdr{Parameters $(n, p, d)$} 
Figure~\ref{fig:additional-p} shows the success probability at inferring the incoming links of nodes on the same type of canonical networks as depicted in Fig.~\ref{fig:canonical-networks}. We
choose nodes the same in-degree but different super-neighboorhod set sizes $p_i$ and experiment with different scalings $\beta$ of the number of cascades $n = 10 \beta d \log p$. We set
the re\-gu\-la\-ri\-zation parameter $\lambda_n$ as a constant factor of $\sqrt{\log (p)/ n}$ as suggested by Theorem~\ref{th:main-result} and, for each node, we used cascades which contained at 
least one node in the super-neighborhood of the node under study. 
We used an exponential transmission model and time window $T = 10$.
As predicted by Theorem~\ref{th:main-result}, very different $p$ values lead to curves that line up with each other quite well.

Figure~\ref{fig:additional-d} shows the success probability at inferring the incoming links of nodes of a hierarchical Kronecker network with equal super neighborhood size ($p_i = 70$) but 
different in-degree ($d_i$) under different scalings $\beta$ of the number of cascades $n = 10 \beta d \log p$ and choose the regularization parameter $\lambda_n$ as a constant factor 
of $\sqrt{\log (p)/ n}$ as suggested by Theorem~\ref{th:main-result}.
We used an exponential transmission model and time window $T = 5$.
As predicted by Theorem~\ref{th:main-result}, in this case, different $d$ values lead to noticeably different curves.

\xhdr{Comparison with \netrate and First-Edge} Figure~\ref{fig:additional-comparison} compares the accuracy of our algorithm, \netrate and First-Edge against number of cascades for 
different type of networks and transmission models.
Our method typically outperforms both competitive methods. We find especially striking the competitive advantage with respect to First-Edge, however, this may be explained
by comparing the sample complexity results for both methods: First-Edge needs $O(N d \log N)$ cascades to achieve a probability of success approaching $1$ in a rate polynomial in the 
number of cascades while our method needs $O(d^3 \log N)$ to achieve a probability of success approaching $1$ in a rate exponential in the number of cascades.